\documentclass[11pt,a4paper,reqno]{amsart}
\usepackage{multirow}
\usepackage{cite}
\usepackage[centertags]{amsmath}
\usepackage{amsfonts}
\usepackage{amssymb}
\usepackage{amsthm}
\usepackage{newlfont}
\usepackage{a4}
\usepackage{enumerate}
\usepackage[pdftex]{graphicx,color}
\theoremstyle{definition}
\newtheorem{defn}{Definition}[section]
\newtheorem{thm}[defn]{Theorem}

\newtheorem{tvr}[defn]{Proposition}
\newtheorem{cor}[defn]{Corollary}
\theoremstyle{remark}
\newtheorem{example}{Example}[section]
\newcommand{\id}{\mathfrak{1}}
\usepackage{bbm}

\newlength{\defbaselineskip}
\newcommand{\setlinespacing}[1]%
           {\setlength{\baselineskip}{#1 \defbaselineskip}}

\renewcommand{\i}{\mathrm{i}}
\newcommand{\map}{\rightarrow}

\newcommand{\q}{\quad}

\renewcommand{\epsilon}{\varepsilon}

\newcommand{\ep}{\varepsilon}
\newcommand{\la}{\lambda}
\newcommand{\al}{\alpha}
\newcommand{\om}{\omega}
\renewcommand{\rho}{\varrho}
\renewcommand{\phi}{\varphi}

\newcommand{\R}{{\mathbb{R}}}

\newcommand{\N}{{\mathbb N}}

\newcommand{\Com}{{\mathbb C}}
\newcommand{\Z}{\mathbb{Z}}
\newcommand{\C}{\mathbb{C}}

\newcommand{\set}[2]{\left\{#1 \, |\, #2 \right\}}
\newcommand{\setb}[2]{\left\{#1 \, \mid\, #2 \right\}}

\newcommand{\abs}[1]{\left\vert#1\right\vert}
\newcommand{\wt}{\widetilde}

\newcommand{\sca}[2]{\langle #1,\, #2\rangle}

\begin{document}

\title[Generalized discrete orbit function transforms]
{Generalized discrete orbit function transforms of affine Weyl groups}

\author[T. Czy\.zycki]{Tomasz Czy\.zycki$^{1}$}
\author[J. Hrivn\'{a}k]{Ji\v{r}\'{i} Hrivn\'{a}k$^{2}$}

\date{\today}
\begin{abstract}\small The affine Weyl groups with their corresponding four types of orbit functions are considered. Two independent admissible shifts, which preserve the symmetries of the weight and the dual weight lattices, are classified. Finite subsets of the shifted weight and the shifted dual weight lattices, which serve as a sampling grid and a set of labels of the orbit functions, respectively, are introduced. The complete sets of discretely orthogonal orbit functions over the sampling grids are found and the corresponding discrete Fourier transforms are formulated. The eight standard one-dimensional discrete cosine and sine transforms form special cases of the presented transforms.   
\end{abstract}

\maketitle
\noindent
$^1$ Institute of Mathematics, University of Bia\l ystok,
Akademicka 2, 15--267 Bia\l ystok, Poland \\
$^2$ Department of Physics, Faculty of Nuclear Sciences and Physical Engineering, Czech Technical University in Prague, B\v{r}ehov\'a~7, CZ--115 19 Prague, Czech Republic
\vspace{10pt}

\noindent
\textit{E-mail:} jiri.hrivnak@fjfi.cvut.cz, tomczyz@math.uwb.edu.pl

\section{Introduction}

This paper aims to extend the results of the discrete Fourier calculus of orbit functions \cite{HP,HMP}. These four types of orbit functions  are induced by the affine Weyl groups of the compact simple Lie groups. Firstly, notions of an admissible shift and an admissible dual shift are introduced and classified. Then two finite subsets of the shifted dual weight lattice and of the shifted weight lattice are introduced --- the first set serves as a sampling grid and the second as a complete set of labels of the discretely orthogonal orbit functions. 

The four types of orbit functions are a generalization of the standard symmetric and antisymmetric orbit sums \cite{Bour}. They can also be viewed as multidimensional generalizations of one-dimensional cosine and sine functions with the symmetry and periodicity determined by the affine Weyl groups \cite{H2}. Indeed, for the case of rank one are the standard cosine and sine functions special cases of orbit functions. For a detailed review of the properties of the symmetric and antisymmetric orbit functions see \cite{KP1,KP2} and the references therein. This paper focuses on discrete Fourier transforms of the orbit functions \cite{HP,HMP,MP2}. 

For the case of rank one, the discrete Fourier transforms from \cite{HP} become one-dimensional discrete cosine and sine transforms known as DCT--I and DST--I \cite{Brit}. The two transforms DCT--I and DST--I constitute only a part of the collection of one-dimensional cosine and sine transforms. The other six most ubiquitous transforms DCT--II, DCT--III, DCT--IV and DST--II, DST--III, DST--IV are obtained by imposing different boundary conditions on the corresponding difference equations \cite{Brit}. The crucial fact is that the resulting grids and the resulting labels of the functions are shifted from their original position depending on the given boundary conditions while preserving their symmetry. This observation serves as a starting point for deriving and generalizing these transforms in the present paper. Besides the standard multidimensional Cartesian product generalizations of DCT's and DST's, other approaches, which also develop multidimensional analogues of the four types of sine and cosine transforms, are based on antisymmetric and symmetric trigonometric functions. 

The (anti)symmetric cosine and sine functions are introduced in \cite{KPtrig} and are directly connected to the four types of orbit functions of the series of the root systems $B_n$ and $C_n$ --- see for instance the three-dimensional relations in \cite{HaHrPa2}. These functions arise as Cartesian products of one-dimensional trigonometric functions which are symmetrized with respect to the permutation group $S_n$. Four types of generalized cosine and sine transforms of both symmetric and antisymmetric types are formulated in \cite{KPtrig} and detailed for two-variable functions in \cite{HP2d,sin}. This approach, however, relies on the validity of the one-dimensional DCT's and DST's and obtained multidimensional shifted discrete grids are subsets of the Cartesian product grids, used in the standard multidimensional version of DCT's and DST's. This paper uses a different approach --- generalized DCT's and DST's are derived independently on their one-dimensional versions and the resulting grids are not subsets of Cartesian product grids.          

The physical motivation of this work stems from widespread use of various types of one-dimensional and multidimensional DCT's and DST's. A textbook case of using one-dimensional DCT's and DST's is description of modes of a beaded string where the type of the transform is determined by the positions of the beads and their boundary conditions. Similar straightforward applications might be expected in two-dimensional and three-dimensional settings. Other more involved applications include general interpolation methods, for instance in chemical physics \cite{YuAnNy},
quantum algorithms \cite{Cao} and quantum communication processes \cite{Bose}.

From a mathematical point of view, the present approach uses the following two types of homomorphisms \cite{MP2,HP,HP2}, 
\begin{itemize}
\item the standard retraction homomorphisms $\psi, \widehat \psi$ of the affine and the dual affine Weyl groups,
\item the four sign homomorphisms $\sigma$, which determine the type of the special function,  
\end{itemize}
and adds two types of homomorphisms related to two shifts $\rho^\vee$ and $ \rho$,
\begin{itemize}
\item the shift and the dual shift homomorphisms  $\widehat{\theta}_{\rho^\vee}, \theta_{\rho}$, which control the affine boundaries conditions,
\item the two $\gamma-$homomorphisms, which combine the previous three types of homomorphisms and control the behavior of the orbit functions on the boundaries.    
\end{itemize}

In Section 2, the notation and pertinent properties concerning the affine Weyl groups and the corresponding lattices are reviewed. The admissible shifts and the dual admissible shifts are classified and the four types of homomorphisms are presented. In Section 3, the four types of orbit functions are recalled and their symmetries, depending on the shifts, are determined. The finite set of grid of points and the labels are introduced and the numbers of their points calculated. In Section 4, the discrete orthogonality of orbit functions is shown and the corresponding discrete Fourier transforms presented.
\section{Pertinent properties of affine Weyl groups}

\subsection{Roots and reflections}\

The notation, established in \cite{HP}, is used. Recall that, to the simple Lie algebra of rank $n$, corresponds the set of simple roots $\Delta=\{\al_1,\dots,\al_n\}$ of  the root system $\Pi$ \cite{BB,Bour,H2}. The set $\Delta$ spans the Euclidean space $\R^n$, with the scalar product denoted by $\sca{\,}{\,}$. The set of simple roots determines partial ordering $\leq$ on $\R^n$ ---  for $\la,\nu \in \R^n$ it holds that  $\nu\leq \la$ if and only if $\la-\nu = k_1\al_1+\dots+ k_n \al_n$ with $k_i \geq 0$ for all $i\in \{1,\dots,n\}$.   The root system $\Pi$ and its set of simple roots $\Delta$ can be defined independently on Lie theory and such sets which correspond to compact simple Lie groups are called crystallographic \cite{H2}. There are two types of sets of simple roots --- the first type with roots of only one length, denoted standardly as $A_n, \, n\geq 1$, $D_n, \, n\geq 4$, $E_6$, $E_7$, $E_8$ and the second type
with two different lengths of roots, denoted $B_n,\, n\geq 3$, $C_n,\, n\geq 2$, $G_2$ and $F_4$. For the second type systems, the set of simple roots consists of short simple roots $\Delta_s$ and long simple roots $\Delta_l$, i.e. the following disjoint decomposition is given,
\begin{equation}\label{sl}
\Delta=\Delta_s\cup\Delta_l.
\end{equation}
The standard objects, related to the set $\Delta\subset \Pi$, are the following \cite{BB,H2}:
\begin{itemize}
\item the highest root $\xi \in \Pi$ with respect to the partial ordering $\leq$ restricted on $\Pi$
\item the marks $m_1,\dots,m_n$ of the highest root $\xi\equiv -\al_0=m_1\al_1+\dots+m_n\al_n$,
\item
the Coxeter number $m=1+m_1+\dots+m_n$,
\item
the root lattice $Q=\Z\al_1+\dots+\Z\al_n $,
\item
the $\Z$-dual lattice to $Q$,
\begin{equation*}
 P^{\vee}=\set{\om^{\vee}\in \R^n}{\sca{\om^{\vee}}{\al}\in\Z,\, \forall \al \in \Delta}=\Z \om_1^{\vee}+\dots +\Z \om_n^{\vee},
 \end{equation*}
with 
\begin{equation}\label{aldom} 
\sca{\al_i}{ \om_j^{\vee}}=\delta_{ij},
\end{equation}
\item
the dual root lattice $Q^{\vee}=\Z \al_1^{\vee}+\dots +\Z \al^{\vee}_n$, where $\al^{\vee}_i=2\al_i/\sca{\al_i}{\al_i}$, $i\in \{1,\dots,n\}$,
\item
the dual marks $m^{\vee}_1, \dots ,m^{\vee}_n$ of the highest dual root $\eta\equiv -\al_0^{\vee}= m_1^{\vee}\al_1^{\vee} + \dots + m_n^{\vee} \al_n^{\vee}$; the marks and the dual marks are summarized in Table 1 in \cite{HP},
\item the $\Z$-dual lattice to $Q^\vee$
\begin{equation*}
 P=\set{\om\in \R^n}{\sca{\om}{\al^{\vee}}\in\Z,\, \forall \al^{\vee} \in Q^\vee}=\Z \om_1+\dots +\Z \om_n,
\end{equation*}
\item
the Cartan matrix $C$ with elements $C_{ij}=\sca{\al_i}{\al^{\vee}_j}$
and with the property
\begin{equation}\label{Cartanm}
 \al^\vee_j = C_{kj} \om_k^{\vee},
\end{equation}
\item the index of connection $c$ of $\Pi$ equal to the  determinant of the Cartan matrix $C$,
\begin{equation}\label{Center}
 c=\det C,
\end{equation}
and determining the orders of the isomorphic quotient groups $P/Q $ and $P^\vee/Q^\vee $,
 \begin{equation*}
 c=|P/Q | =|P^\vee/Q^\vee |.
\end{equation*}
\end{itemize}

The $n$ reflections $r_\al$, $\al\in\Delta$ in $(n-1)$-dimensional mirrors orthogonal to simple roots intersecting at the origin are given explicitly for $a\in\R^n$ by
\begin{equation}\label{refl}
r_{\al}a=a-\sca{\al}{a}\al^\vee. 
\end{equation}
and the affine reflection $r_0$ with respect to the highest root $\xi$ is given by
\begin{equation}\label{aWeyl}
r_0 a=r_\xi a + \frac{2\xi}{\sca{\xi}{\xi}}\,,\qquad
r_{\xi}a=a-\frac{2\sca{a}{\xi} }{\sca{\xi}{\xi}}\xi\,,\qquad a\in\R^n\,.
\end{equation}
The set of reflections $r_1\equiv r_{\al_1}, \, \dots, r_n\equiv r_{\al_n}$, together with the affine reflection $r_0$,  is denoted by $R$,
\begin{equation}\label{R}
R=\{ r_0,r_1,\dots,r_n \}.
\end{equation}
The dual affine reflection $r_0^{\vee}$, with respect to the dual highest root $\eta$, is given by
\begin{equation}\label{daWeyl}
r_0^{\vee} a=r_{\eta} a + \frac{2\eta}{\sca{\eta}{\eta}}, \q r_{\eta}a=a-\frac{2\sca{a}{\eta} }{\sca{\eta}{\eta}}\eta,\q a\in\R^n.
\end{equation}
The set of reflections $r^\vee_1\equiv r_{\al_1}, \, \dots, r^\vee_n\equiv r_{\al_n}$, together with the dual affine reflection $r^\vee_0$ is denoted by $R^\vee$,
\begin{equation*}
R^{\vee}=\{ r_0^{\vee},r^\vee_1,\dots,r^\vee_n \}.
\end{equation*}

\subsection{Weyl group and affine Weyl group}\

The Weyl group $W$ is generated by $n$ reflections $r_\al$, $\al\in\Delta$. The set $R$ of $n+1$ generators \eqref{R} generates the affine Weyl group $W^{\mathrm{aff}}$. Except for the case $A_1$, the Coxeter presentation of  $W^{\mathrm{aff}}$ is of the form
\begin{equation}\label{CoxW}
	W^{\mathrm{aff}}=\langle R \mid  (r_i r_j)^{m_{ij}}=1 \rangle ,
\end{equation}
where the numbers $m_{ij}$, $i,j\in\{0,1,\dots, n\}$ are given by the extended Coxeter-Dynkin diagrams --- see e.~g.~\cite{HP}. It holds that $m_{ii}=1$ and if the $i$th and $j$th nodes in the diagram are not connected then $m_{ij}=2$; otherwise the single, double or triple vertices between the nodes indicate  $m_{ij}$ equal to $3, 4$ or $6$, respectively.

Moreover, the affine Weyl group $W^{\mathrm{aff}}$ is the semidirect product of the Abelian group of translations $T(Q^\vee)$ by shifts from $Q^\vee$ and of the Weyl group~$W$,
\begin{equation*}
 W^{\mathrm{aff}}= T(Q^\vee) \rtimes W.
\end{equation*}
Thus, for any $w^{\mathrm{aff}}\in W^{\mathrm{aff}}$, there exist a unique $w\in W$ and a unique shift $T(q^{\vee})$ such that $w^{\mathrm{aff}}=T(q^{\vee})w$.
Taking any $w^{\mathrm{aff}}=T(q^\vee)w\in {W}^{\mathrm{aff}}$, the retraction homomorphism $\psi:{W}^{\mathrm{aff}}\map W$  and the mapping $\tau:{W}^{\mathrm{aff}}\map Q^\vee$ are given by
\begin{align}
\psi(w^{\mathrm{aff}}) &=w, \label{ret}\\
\tau(w^{\mathrm{aff}}) &=q^\vee \label{retq}.
\end{align}

The fundamental domain $F$ of $W^{\mathrm{aff}}$, which consists of precisely one point of each $W^{\mathrm{aff}}$-orbit, is the convex hull of the points $\left\{ 0, \frac{\om^{\vee}_1}{m_1},\dots,\frac{\om^{\vee}_n}{m_n} \right\}$. Considering $n+1$ real parameters $y_0,\dots, y_n\geq 0$, we have
\begin{align}
F &=\setb{y_1\om^{\vee}_1+\dots+y_n\om^{\vee}_n}{y_0+y_1 m_1+\dots+y_n m_n=1  }. \label{deffun}
\end{align}
Let us denote the isotropy subgroup of a point $a\in\R^n$ and its order by
\begin{equation*}
\mathrm{Stab}_{W^{\mathrm{aff}}}(a) = \setb{w^{\mathrm{aff}}\in W^{\mathrm{aff}}}{w^{\mathrm{aff}}a=a},\q h(a)=|\mathrm{Stab}_{W^{\mathrm{aff}}}(a)|,
\end{equation*}
and define a function $\ep:\R^n\map\N$ by the relation
\begin{equation}\label{epR}
\ep(a)=\frac{|W|}{h(a)}.
\end{equation}
Since the stabilizers $\mathrm{Stab}_{W^{\mathrm{aff}}}(a) $ and $\mathrm{Stab}_{W^{\mathrm{aff}}}(w^{\mathrm{aff}}a) $ are for any $w^{\mathrm{aff}}\in W^{\mathrm{aff}}$ conjugated, one obtains that
\begin{equation}\label{epshift}
\ep(a)=\ep(w^{\mathrm{aff}}a),\q w^{\mathrm{aff}}\in W^{\mathrm{aff}}.
\end{equation}

Recall that the stabilizer $\mathrm{Stab}_{W^{\mathrm{aff}}}(a)$
of a point $a=y_1\om^{\vee}_1+\dots+y_n\om^{\vee}_n\in F$ is trivial, $\mathrm{Stab}_{W^{\mathrm{aff}}}(a)=1$ if the point $a$ is in the interior of $F$, $a\in \mathrm{int}(F)$. Otherwise the group $\mathrm{Stab}_{W^{\mathrm{aff}}}(a)$
is generated by such $r_i$ for which $y_i=0$, $i=0,\dots,n$.

Considering the standard action of $W$ on the torus $\R^n/Q^{\vee}$, we denote for $x\in \R^n/Q^{\vee}$ the isotropy group by $\mathrm{Stab} (x)$
and the orbit and its order by
\begin{equation*}
W x=\set{wx\in \R^n/Q^{\vee} }{w\in W},\q \wt \ep(x)\equiv |Wx|.
\end{equation*}
Recall the following three properties from Proposition 2.2 in \cite{HP} of the action of $W$ on the torus $\R^n/Q^{\vee}$:
\begin{enumerate}
\item For any $x\in \R^n/Q^{\vee}$, there exists $x'\in F \cap \R^n/Q^{\vee} $ and $w\in W$ such that
\begin{equation}\label{rfun1}
 x=wx'.
\end{equation}
\item If $x,x'\in F \cap \R^n/Q^{\vee} $ and $x'=wx$, $w\in W$, then
\begin{equation}\label{rfun2}
 x'=x=wx.
\end{equation}
\item If $x\in F \cap \R^n/Q^{\vee} $, i.e. $x=a+Q^{\vee}$, $a\in F$, then $\psi (\mathrm{Stab}_{W^{\mathrm{aff}}}(a))=\mathrm{Stab}(x)$ and
\begin{equation}\label{rfunstab}
\mathrm{Stab} (x) \cong \mathrm{Stab}_{W^{\mathrm{aff}}}(a).
\end{equation}
\end{enumerate}
From \eqref{rfunstab} we obtain that for $x=a+Q^{\vee}$, $a\in F$ it holds that
\begin{equation}\label{ept}
	\ep(a)= \wt\ep(x).
\end{equation}
Note that instead of $\wt\ep(x)$, the symbol $\ep(x)$ is used for $|Wx|$, $x\in F\cap\R^n/Q^{\vee} $ in \cite{HP,HMP}. The calculation procedure of the coefficients $\ep(x)$ is detailed in \S 3.7 in \cite{HP}.

\subsection{Admissible shifts}\

For the development of the discrete Fourier calculus is crucial the notion of certain lattices invariant under the action of the Weyl group $W$. In \cite{HP,HMP} is formulated the discrete Fourier calculus on the fragment of the refined $W-$invariant lattice $P^\vee$. Considering any vector $\rho^\vee \in \R^n$, we call $\rho^\vee$ an admissible shift if the shifted lattice $\rho^\vee+P^\vee$ is still $W-$invariant, i.e.
\begin{equation}\label{rhoa}
W(\rho^\vee+P^\vee)=\rho^\vee+P^\vee.
\end{equation}
If $\rho^\vee\in P^\vee$ then the resulting lattice $\rho^\vee+P^\vee=P^\vee$ does not change --- such admissible shifts are called trivial. Also any two admissible shifts $\rho_1^\vee,\,\rho_2^\vee $ which would differ by a trivial shift, i.e. $\rho_1^\vee-\rho_2^\vee \in P^\vee$ lead to the same resulting shifted lattice and are in this sense equivalent. Thus, in the following, we classify admissible shifts up to this equivalence. The following proposition significantly simplifies the classification of admissible shifts. 
\begin{tvr}\label{Wreq}
Let $\rho^\vee \in \R^n$. Then the following statements are equivalent.
\begin{enumerate}
\item $\rho^\vee$ is an admissible shift.
\item $\rho^\vee- W\rho^\vee\subset P^\vee$.
\item For all $i\in \{1,\dots,n\}$ it holds that 
\begin{equation}\label{rirho}
\rho^\vee- r_i\rho^\vee\in P^\vee.\end{equation}
\end{enumerate}
\end{tvr}
\begin{proof}
$(1)\Rightarrow(2)$: If $\rho^\vee$ is admissible then for every $w\in W$ and every $p_1^\vee\in P^\vee$ there exists $p_2^\vee\in P^\vee$ such that $ w(\rho^\vee+p_1^\vee)= \rho^\vee + p_2^\vee $. Then $\rho^\vee - w\rho^\vee = wp_1^\vee-p_2^\vee \in P^\vee$.

$(2)\Rightarrow(3)$: If for every $w\in W$ it holds that $\rho^\vee- w\rho^\vee\in P^\vee$, thus this equality is also valid for all $r_i\in W$. 

$(3)\Rightarrow(1)$: Any $w\in W$ can be expressed as a product of generators, i.e. there exist indices $i_1,i_2,\dots,i_s\in \{1,\dots,n\}$ such that $w=r_{i_1}r_{i_2}\dots r_{i_s}$. Thus, we have from the assumption that there exist vectors $p^\vee_{i_1},\dots,p^\vee_{i_s}\in P^\vee$ such that $\rho^\vee- r_{i_1}\rho^\vee=p^\vee_{i_1}$, $\rho^\vee- r_{i_2}\rho^\vee=p^\vee_{i_2}$, $\dots$, $\rho^\vee- r_{i_s}\rho^\vee=p^\vee_{i_s}$. Then we derive
\begin{align*}
\rho^\vee-w\rho^\vee&=\rho^\vee-r_{i_1}r_{i_2}\dots r_{i_s}\rho^\vee=\rho^\vee- r_{i_1}r_{i_2}\dots r_{i_{s-1}}(\rho^\vee - p^\vee_{i_s}) \\
&= \rho^\vee-r_{i_1}r_{i_2}\dots r_{i_{s-2}}(\rho^\vee - p^\vee_{i_{s-1}})  +r_{i_1}r_{i_2}\dots r_{i_{s-1}}p^\vee_{i_s}\\
&= p^\vee_{i_{1}} + r_{i_{1}}p^\vee_{i_{2}}+r_{i_{1}}r_{i_{2}}p^\vee_{i_{3}}+\dots+ r_{i_1}r_{i_2}\dots r_{i_{s-1}}p^\vee_{i_s}.
\end{align*}
Denoting $\wt p^\vee=p^\vee_{i_{1}} + r_{i_{1}}p^\vee_{i_{2}}+r_{i_{1}}r_{i_{2}}p^\vee_{i_{3}}+\dots+ r_{i_1}r_{i_2}\dots r_{i_{s-1}}p^\vee_{i_s}$ we have that $\wt p^\vee\in P^\vee$ since $P^\vee$ is $W-$invariant.
Thus for all $w\in W$ there exists $\wt p^{\vee}\in P^\vee$ such that $\rho^\vee -w\rho^\vee = \wt p^\vee$ and for all $p^\vee\in P^\vee$ it holds that
$$w(\rho^\vee + p^\vee)= -\wt p^\vee +\rho^\vee + wp^\vee, $$
i.e. there exists $\wt p^{'\vee} = wp^\vee - \wt p^\vee\in P^\vee$ such that $w(\rho^\vee + p^\vee)=\rho^\vee +\wt p^{'\vee}$ and $\rho^\vee$ is admissible.
\end{proof}

Analyzing the condition \eqref{rirho}, we note that it is advantageous to consider a shift $\rho^\vee$ -- up to equivalence -- in $\om^\vee-$basis, i.e. 
\begin{equation}\label{shiftt}
\rho^\vee = y_j \om^\vee _j, \q y_j \in \langle 0,\, 1). \end{equation}
Using the relations \eqref{aldom}, \eqref{Cartanm} and \eqref{refl}, we calculate that
\begin{equation*}
\rho^\vee- r_i\rho^\vee=y_iC_{ki}\om^\vee_k .
\end{equation*}
Thus from \eqref{rirho}, in order to $\rho^\vee$ be admissible, it has to hold $y_iC_{ki}\in\Z$ for all $i\in \{1,\dots,n\}$. Let us define the numbers 
\begin{equation}\label{gcddef}
d_i=\gcd (C_{1i},C_{2i},\dots,C_{ni} )  
\end{equation}
i.e. the integer $d_i$ is the greatest common divisor of the $i-$th column of the Cartan matrix. Each $d_i$ can be expressed as the integer combination of $C_{1i},C_{2i},\dots,C_{ni}$ in the form 
\begin{equation*}
d_i=k_{1i}C_{1i}+k_{2i}C_{2i}+\dots + k_{1i}C_{1i}, \q k_{ki}\in \Z
\end{equation*}
and we obtain from $y_iC_{ki}\in\Z$ that $y_id_i \in \Z$. Conversely from $y_id_i \in \Z$ and $C_{ki}/d_i\in\Z$ we have that $y_iC_{ki}\in\Z$ and thus we conclude: 
\begin{cor}
A shift of the form \eqref{shiftt} is admissible if and only if for all $i\in \{1,\dots,n\}$ it holds that $d_iy_i\in \Z$ where the numbers $d_i$ are defined by \eqref{gcddef}.
\end{cor}
The Cartan matrices taken from e.g. \cite{BMP} are examined and cases for which a non-trivial admissible shift exists, i.e. those with some $d_i=2$ and yielding  $y_i=1/2$, are singled out. It appears that non-trivial admissible shifts exist for the cases $A_1$, $C_2$ and $B_n,\, n\geq 3$. These are summarized in Table \ref{Tab:shifts}.

\begin{table}
{\small
\begin{tabular}{|c|c|c|}\hline
      & $\rho^\vee$ & $\rho$  \\[1pt] \hline\hline
$A_1$ & $ \frac{1}{2}\om^\vee_1$ & $ \frac{1}{2}\om_1 $  \\[3pt]
$C_2$ & $ \frac{1}{2}\om^\vee_1$ & $ \frac{1}{2}\om_2 $ \\[3pt]
$B_n,\, n\geq 3$ & $ \frac{1}{2}\om^\vee_n$  & $-$ \\[3pt]
$C_n,\, n\geq 3$ & $-$  & $ \frac{1}{2}\om_n$ \\[1pt]
\hline
\end{tabular}}\medskip
\caption{Non-trivial admissible shifts $\rho^\vee$ and admissible dual shifts $\rho$ of affine Weyl groups with the standard realization of the root systems \cite{BMP}.} \label{Tab:shifts}
\end{table}

\subsection{Dual affine Weyl group}\

The dual affine Weyl group $\widehat{W}^{\mathrm{aff}}$ is generated by the set $R^\vee$ and, except for the case $A_1$, the Coxeter presentation of  $\widehat{W}^{\mathrm{aff}}$ is of the form
\begin{equation*}
	\widehat{W}^{\mathrm{aff}}=\langle R^\vee \mid  (r^\vee_i r^\vee_j)^{m^\vee_{ij}}=1 \rangle .
\end{equation*}
The numbers $m^\vee_{ij}$, $i,j\in\{0,1,\dots, n\}$ are deduced, following the identical rules as for ${W}^{\mathrm{aff}}$, from the dual extended Coxeter-Dynkin diagrams \cite{HP}.

The dual affine Weyl group $\widehat{W}^{\mathrm{aff}}$ is a semidirect product of the group of shifts $T(Q)$ and the Weyl group $W$
\begin{equation*}
 \widehat{W}^{\mathrm{aff}}= T(Q) \rtimes W.
\end{equation*}
For any $w^{\mathrm{aff}}\in\widehat{W}^{\mathrm{aff}}$ there exist a unique element $w\in W$ and a unique shift $T(q)$, $q\in Q$ such that $w^{\mathrm{aff}}=T(q)w$.
Taking any $w^{\mathrm{aff}}=T(q)w\in \widehat{W}^{\mathrm{aff}}$, the dual retraction homomorphism $\widehat\psi:\widehat{W}^{\mathrm{aff}}\map W$  and the mapping $\widehat\tau:\widehat{W}^{\mathrm{aff}}\map Q$ are given by
\begin{align}
\widehat\psi(w^{\mathrm{aff}}) &=w, \label{retd}\\
\widehat\tau(w^{\mathrm{aff}}) &=q \label{retdq}.
\end{align}
The dual fundamental domain $F^\vee$ of $\widehat{W}^{\mathrm{aff}}$ is the convex hull of vertices $\left\{ 0, \frac{\om_1}{m^{\vee}_1},\dots,\frac{\om_n}{m^{\vee}_n} \right\}$. Considering $n+1$ real parameters $z_0,\dots, z_n\geq 0$, we have
\begin{align}
F^\vee &=\setb{z_1\om_1+\dots+z_n\om_n}{z_0+z_1 m_1^{\vee}+\dots+z_n m^{\vee}_n=1  }.\label{deffund} 
\end{align}

Let us denote the isotropy subgroup of a point $b\in\R^n$ by $\mathrm{Stab}_{\widehat{W}^{\mathrm{aff}}}(b)$
and define for any $M\in \N$ a function $h^\vee_M:\R^n\map\N$ by the relation
\begin{equation}\label{hMb}
h^\vee_M(b)=\left|\mathrm{Stab}_{\widehat{W}^{\mathrm{aff}}}\left(\frac{b}{M}\right)\right|.
\end{equation}
Recall that for a point $b=z_1\om_1+\dots+z_n\om_n\in F^\vee$ such that $z_0+z_1 m_1^{\vee}+\dots+z_n m^{\vee}_n=1$ is the isotropy group
trivial, $ \mathrm{Stab}_{\widehat{W}^{\mathrm{aff}}}(b)=1$, if $b\in \mathrm{int}(F^\vee)$, i.e. all $z_i>0$, $i=0,\dots,n$. Otherwise the group $\mathrm{Stab}_{\widehat{W}^{\mathrm{aff}}}(b)$
is generated by such $r^{\vee}_i$ for which $z_i=0$, $i=0,\dots,n$.

From $F^\vee$ being a fundamental region follows that
\begin{enumerate}
\item For any $b\in \R^n$ there exists $b'\in F^\vee$, $w\in W$ and $q\in Q$ such that
\begin{equation}\label{ddfun1}
 b=wb'+q.
\end{equation}
\item If $b,b'\in F^\vee$ and $b'=w^{\mathrm{aff}}b$, $w^{\mathrm{aff}}\in \widehat{W}^{\mathrm{aff}}$ then $b=b'$, i.e. if there exist $w\in W$ and $q\in Q$ such that $b'=wb+q$ then
\begin{equation}\label{dfun2}
 b'=b=wb+q.
\end{equation}
\end{enumerate}

Considering a natural action of $W$ on the quotient group $\R^n/MQ$, we denote for $\la \in \R^n/MQ$ the isotropy group and its order by
\begin{equation}\label{hla}
\mathrm{Stab}^{\vee} (\la)=\set{w\in W}{w\la=\la},\q \wt h^{\vee}_M(\la)\equiv |\mathrm{Stab}^{\vee} (\la)|.
\end{equation}
Recall the following property from Proposition 3.6 in \cite{HP} of the action of $W$ on the quotient group $\R^n/MQ$.
If $\la\in M F^\vee \cap \R^n/MQ $, i.e. $\la=b+MQ$, $b\in MF^\vee$, then $\widehat\psi (\mathrm{Stab}_{\widehat{W}^{\mathrm{aff}}}(b/M))=  \mathrm{Stab}^{\vee} (\la)$ and
\begin{equation}\label{rfunstab2}
\mathrm{Stab}^{\vee} (\la) \cong \mathrm{Stab}_{\widehat{W}^{\mathrm{aff}}}(b/M).
\end{equation}
From \eqref{rfunstab2} we obtain that for $\la=b+MQ$, $b\in F^\vee$ it holds that
\begin{equation}
h^\vee_M(b)= \wt h^{\vee}_M(\la).
\end{equation}
Note that instead of $\wt h^{\vee}_M(\la)$, the symbol $h^\vee_\la$ is used for $|\mathrm{Stab}^{\vee} (\la)|$, $\la \in\R^n/MQ $ in \cite{HP,HMP}. The calculation procedure of the coefficients $h^\vee_\la$ is detailed in \S 3.7 in \cite{HP}.

\subsection{Admissible dual shifts}\

The second key ingredient of the discrete Fourier calculus is a lattice, invariant under the action of the Weyl group $W$, which will label the sets of orthogonal functions. In \cite{HP,HMP} is formulated the discrete Fourier calculus with special functions labeled by the labels from the $W-$invariant weight lattice $P$. Considering any vector $\rho \in \R^n$, we call $\rho$ an admissible dual shift if the shifted lattice $\rho+P$ is still $W-$invariant, i.e.
\begin{equation*}
W(\rho+P)=\rho+P.
\end{equation*}
Similarly to the shifts, if $\rho\in P$ then the resulting lattice $\rho+P=P$ does not change --- such admissible dual shifts are called trivial. Also any two dual admissible shifts $\rho_1,\,\rho_2 $ which would differ by a trivial shift, i.e. $\rho_1-\rho_2 \in P$ lead to the same resulting shifted lattice and are equivalent. Rephrasing of Proposition \ref{Wreq} leads to the classification of the admissible dual shifts. 
\begin{tvr}\label{Wdreq}
Let $\rho \in \R^n$. Then the following statements are equivalent.
\begin{enumerate}
\item $\rho$ is an admissible dual shift.
\item $\rho- W\rho\subset P$.
\item For all $i\in \{1,\dots,n\}$ it holds that 
\begin{equation}\label{rirhod}
\rho- r_i\rho\in P.\end{equation}
\end{enumerate}
\end{tvr}

Analyzing similarly the condition \eqref{rirhod} and considering a dual shift $\rho$ up to equivalence in $\om-$basis, i.e. 
\begin{equation}\label{shifttd}
\rho = y_j \om _j, \q y_j \in \langle 0,\, 1). 
\end{equation}
we calculate that
\begin{equation*}
\rho- r_i\rho=y_iC_{ik}\om_k .
\end{equation*}
Defining the numbers 
\begin{equation}\label{gcddefd}
d'_i=\gcd (C_{i1},C_{i2},\dots,C_{in} )  
\end{equation}
i.e. the integer $d'_i$ is the greatest common divisor of the $i-$th row of the Cartan matrix, we again conclude: 
\begin{cor}
A dual shift of the form \eqref{shifttd} is admissible if and only if for all $i\in \{1,\dots,n\}$ it holds that $d'_iy_i\in Z$ where the numbers $d'_i$ are defined by \eqref{gcddefd}.
\end{cor}
The Cartan matrices are repeatedly examined and cases for which a non-trivial dual admissible shift exists, i.e. those with some $d'_i=2$ and yielding $y_i=1/2$, are singled out. It appears that non-trivial dual admissible shifts exist for the cases $A_1$, $C_2$ and $C_n,\, n\geq 3$ only. These are summarized in Table \ref{Tab:shifts}.

\subsection{Shift homomorphisms}\

An admissible shift $\rho^\vee$ induces a homomorphism from the dual affine Weyl group to the multiplicative two-element group $\{\pm 1\}$. This 'shift' homomorphism $\widehat{\theta}_{\rho^\vee}: \widehat{W}^{\mathrm{aff}}\map \{\pm 1\} $ is defined for $w^{\mathrm{aff}}\in\widehat{W}^{\mathrm{aff}}$, using the mapping \eqref{retdq}, via the equation
\begin{equation}\label{shifthom}
\widehat{\theta}_{\rho^\vee} (w^{\mathrm{aff}})=e^{2\pi i \sca{ \widehat{\tau}(w^{\mathrm{aff}})}{\rho^\vee}}.
\end{equation}
For trivial admissible shifts we obtain a trivial homomorphism $\widehat{\theta}_{\rho^\vee} (w^{\mathrm{aff}})=1$. Since all non-trivial admissible shifts $\rho^\vee$ are from $\frac{1}{2} P^\vee$, the map \eqref{shifthom} indeed maps to $\{\pm 1\}$.  The key point to see that  \eqref{shifthom} defines a homomorphism is that the equality 
\begin{equation}\label{shiftexp}
e^{2\pi i \sca{wq}{\rho^\vee}}=e^{2\pi i \sca{q}{\rho^\vee}}, \q q\in Q, \q w\in W,
\end{equation}
which is equivalent to $\sca{q}{\rho^\vee-w^{-1}\rho^\vee}\in \Z$, is valid because $\rho^\vee-w^{-1}\rho^\vee\in P^\vee$  is the second statement of Proposition~\ref{Wreq} and $P^\vee$ is $\Z-$dual to $Q$.
Considering two dual affine Weyl group elements $w_1^{\mathrm{aff}}, \,w_2^{\mathrm{aff}} \in \widehat{W}^{\mathrm{aff}}$ of the form $w_1^{\mathrm{aff}}a=w_1 a + q_1$,  $w_2^{\mathrm{aff}}a=w_2 a + q_2$ with $w_1,\,w_2 \in W$ and $q_1,\, q_2 \in Q$ we calculate that $\widehat{\tau}(w_1^{\mathrm{aff}}w_2^{\mathrm{aff}})=w_1q_2+q_1$ and thus
\begin{align*}
\widehat{\theta}_{\rho^\vee} (w_1^{\mathrm{aff}}w_2^{\mathrm{aff}})=&e^{2\pi i \sca{ \widehat{\tau}(w_1^{\mathrm{aff}}w_2^{\mathrm{aff}})}{\rho^\vee}}=e^{2\pi i \sca{w_1q_2+q_1}{\rho^\vee}}=e^{2\pi i \sca{w_1q_2}{\rho^\vee}}e^{2\pi i \sca{q_1}{\rho^\vee}}=e^{2\pi i \sca{q_2}{\rho^\vee}}e^{2\pi i \sca{q_1}{\rho^\vee}} \\ =& \widehat{\theta}_{\rho^\vee} (w_1^{\mathrm{aff}})\widehat{\theta}_{\rho^\vee} (w_2^{\mathrm{aff}}).
\end{align*}
Since the mapping is \eqref{shifthom} a homomorphism, its value on any element of  $\widehat{W}^{\mathrm{aff}}$ is determined by the values of  \eqref{shifthom} on generators from $R^\vee$. For any reflection $r^\vee_1,\dots,r^\vee_n$ it trivially holds that $ \widehat{\tau}(r^\vee_i)=0$ and thus $\widehat{\theta}_{\rho^\vee} (r^\vee_i)=1$ for $i\in \{1,\dots,n\}$. Evaluating the number $\widehat{\theta}_{\rho^\vee} (r_0^\vee)$ one needs to take into account explicitly the non-trivial admissible shifts from Table  \ref{Tab:shifts} and the dual highest roots $\eta$ from \cite{HP}. It appears that the result for all cases is that for non-trivial admissible shifts it holds 
\begin{equation}\label{shiftd}
\frac{2\sca{\eta}{\rho^\vee}}{\sca{\eta}{\eta}}=\frac{1}{2}.
\end{equation}
Then from \eqref{daWeyl} and \eqref{shiftd} we obtain 
\begin{equation*}
\widehat{\theta}_{\rho^\vee} (r_0^\vee)=e^{2\pi i \sca{ \widehat{\tau}(r_0^\vee)}{\rho^\vee}}=e^{2\pi i \frac{2\sca{\eta}{\rho^\vee}}{\sca{\eta}{\eta}}}=-1.
\end{equation*}
Summarizing the results, we conclude with the following proposition.
\begin{tvr}\label{shifthomprop}
The map \eqref{shifthom} is for any admissible shift a homomorphism  $\widehat{\theta}_{\rho^\vee}: \widehat{W}^{\mathrm{aff}}\map \{\pm 1\} $ and for any non-trivial admissible shift $\rho^\vee$ are the values on the generators $R^\vee$ of $\widehat{W}^{\mathrm{aff}}$ given as 
\begin{equation}\label{dthetahom}
\widehat{\theta}_{\rho^\vee}(r^\vee_i)=\begin{cases} 1,\quad i\in \{1,\dots,n\}  \\ -1,\quad i=0.\end{cases}
\end{equation}
\end{tvr}
Similarly to the shift homomorphism \eqref{shifthom} we define for the dual admissible shift $\rho$ the dual shift homomorphism $\theta_{\rho}: W^{\mathrm{aff}}\map \{\pm 1\}$ --- for any $w^{\mathrm{aff}}\in{W}^{\mathrm{aff}}$, using the mapping \eqref{retq}, via the equation
\begin{equation}\label{shiftdhom}
\theta_{\rho} (w^{\mathrm{aff}})=e^{2\pi i \sca{ {\tau}(w^{\mathrm{aff}})}{\rho}}.
\end{equation} 
Analogous relation to \eqref{shiftexp} is also valid
\begin{equation}\label{shiftdexp}
e^{2\pi i \sca{wq^\vee}{\rho}}=e^{2\pi i \sca{q^\vee}{\rho}}, \q q^\vee\in Q^\vee, \q w\in W
\end{equation}
and also
 \begin{equation}\label{shiftdd}
\frac{2\sca{\xi}{\rho}}{\sca{\xi}{\xi}}=\frac{1}{2}.
\end{equation}
Thus we conclude with the following proposition.
\begin{tvr}\label{shiftdhomprop}
The map \eqref{shiftdhom} is for any admissible dual shift a homomorphism  ${\theta}_{\rho}: {W}^{\mathrm{aff}}\map \{\pm 1\} $ and for any non-trivial admissible dual shift $\rho$ are the values on the generators $R$ of ${W}^{\mathrm{aff}}$ given as 
\begin{equation}\label{thetahom}
{\theta}_{\rho}(r_i)=\begin{cases} 1,\quad i\in \{1,\dots,n\}  \\ -1,\quad i=0.\end{cases}
\end{equation}
\end{tvr}

Note that, excluding the case $A_1$,  any homomorphism  ${\theta}: {W}^{\mathrm{aff}}\map \{\pm 1\}$ has to be compatible with the generator relations~\eqref{CoxW} from the Coxeter presentation of ${W}^{\mathrm{aff}}$, $$(\theta(r_i) \theta( r_j))^{m_{ij}}=1, \q i,j\in \{0,\dots,n\}.$$
Thus, except for the admissible cases $C_n, n\geq 2$, a homomorphism ${\theta}_{\rho}$ defined on generators $R$ via \eqref{thetahom} cannot exist  --- the admissible cases are the only cases where the zero vertex of the extended Coxeter-Dynkin diagram is not connected to the rest of the diagram by a single vertex, i.e. $m_{0i}, i\in \{1,\dots,n\}$ are all even. Similarly, analyzing the dual extended Coxeter-Dynkin diagrams, we observe that the admissible cases $C_2$ and $B_n,\, n\geq 3$ are the only ones allowing existence of a homomorphism of the form \eqref{dthetahom}. Since the generators $r_0, r_1$ and the dual generators $r^\vee_0, r^\vee_1$of $A_1$ do not obey any relation between them, the case  $A_1$ indeed admits both homomorphisms of the forms \eqref{thetahom} and \eqref{dthetahom}.

\subsection{Sign homomorphisms and $\gamma-$homomorphisms}\

Up to four classes of orbit functions are obtained using 'sign' homomorphisms $\sigma:W\rightarrow\{\pm1 \}$ -- see e.g. \cite{HP,HMP}. These homomorphisms and their values on any $w\in W$ are given as products of values of generators $r_\al,\,\al\in\Delta$. The following two choices of homomorphism values of generators $r_\al,\,\al\in\Delta$ lead to the trivial and determinant homomorphisms
\begin{align}
\id(r_\al)&=1 \label{ghomid} \\
\sigma^e(r_\al)&=-1, \label{ghome}
\end{align}
yielding for any $w\in W$
\begin{align}
\id(w)&=1 \label{idd} \\
\sigma^e(w)&=\det w. \label{parity}
\end{align}
For root systems with two different lengths of roots, there are two other available choices. Using the decomposition \eqref{sl}, these two new homomorphisms are given as follows \cite{HMP}:
\begin{align}
\sigma^s(r_\al)&=\begin{cases} 1,\quad \al\in \Delta_l  \\ -1,\quad \al\in \Delta_s\end{cases} \label{sigmas}\\
\sigma^l(r_\al)&=\begin{cases} 1,\quad \al\in \Delta_s  \\ -1,\quad \al\in \Delta_l.\end{cases} \label{sigmal}
\end{align}
From \cite{HMP} we also have the values for the reflections \eqref{aWeyl} and \eqref{daWeyl} from $W,$
\begin{align}
\sigma^s(r_\xi)&=1,\q \sigma^l(r_\xi)=-1,\label{slxi}\\
\sigma^s(r_\eta)&=-1,\q \sigma^l(r_\eta)=1.\label{sleta}
\end{align}

At this point we have all three key homomorphisms ready --  the retraction homomorphism $\psi$ , the dual shift homomorphism ${\theta}_{\rho}$ and up to four sign homomorphisms $\sigma$ given by  \eqref{ret}, \eqref{shiftdhom} and \eqref{idd}, \eqref{parity}, \eqref{sigmas}, \eqref{sigmal}, respectively. We use these three homomorphisms to create the fourth and most ubiquitous homomorphism $\gamma^\sigma_\rho: {W}^{\mathrm{aff}}\map \{\pm 1\} $ given by
\begin{equation}\label{gammah}
 \gamma^\sigma_\rho (w^{\mathrm{aff}})= {\theta}_{\rho}(w^{\mathrm{aff}})\cdot [\sigma \circ \psi (w^{\mathrm{aff}})], \q w^{\mathrm{aff}} \in {W}^{\mathrm{aff}}.
\end{equation}
Note that, since for the affine reflection $r_0\in R$ it holds that $\psi(r_0)=r_\xi$, we have from \eqref{slxi} that 
\begin{equation}\label{sigmar0}
 \sigma^s \circ \psi (r_0)=1, \q \sigma^l \circ \psi (r_0)=-1.
\end{equation}
The values of the $\gamma^\sigma_\rho-$homomorphism on any $w^{\mathrm{aff}} \in {W}^{\mathrm{aff}}$ are determined by its values on the generators from $R$. Putting together the values from  \eqref{thetahom}, \eqref{ghomid} --- \eqref{sleta} and \eqref{sigmar0}, we summarize the values of  $\gamma^\sigma_\rho (r),\,r\in R$ for a trivial $0-$shift and the non-trivial admissible dual shift $\rho$ in Table~\ref{Tab:hom}.

\begin{table}
{\small \renewcommand{\arraystretch}{1.3}
\begin{tabular}{|c|c|c|c|c||c|c|c|c|}\hline 
                                    & $\gamma^{\id}_0$ & $\gamma^{\sigma^e}_0$  &  $\gamma^{\sigma^s}_0$ & $\gamma^{\sigma^l}_0$ & $\gamma^{\id}_\rho$  &  $\gamma^{\sigma^e}_\rho$ & $\gamma^{\sigma^s}_\rho$  &  $\gamma^{\sigma^l}_\rho$\\[1pt] \hline\hline
$r_0$                           & $ 1$ & $ -1$ & $ 1$ & $ -1$ & $ -1$ & $ 1$ & $ -1$ & $ 1$ \\[3pt]
$r_\al,\,\al\in\Delta_l$ & $ 1$ & $ -1$ &$ 1$ & $ -1$ & $ 1$ & $ -1$ & $ 1$ & $ -1$ \\[3pt]
$r_\al,\,\al\in\Delta_s$ & $ 1$ & $ -1$ & $ -1$ & $ 1$ & $ 1$ & $ -1$ & $ -1$ & $ 1$  \\[3pt]
\hline
\end{tabular}}\medskip
\caption{The values of  $\gamma^\sigma_\rho (r),\,r\in R$ for a trivial $0-$shift, the non-trivial admissible dual shift $\rho$ and four sign homomorphisms  $\sigma$.} \label{Tab:hom}
\end{table}

We use four sign homomorphisms $\sigma$, the dual retraction homomorphism \eqref{retd} and the shift homomorphism \eqref{shifthom} to define a dual version of the  $\gamma^\sigma_\rho-$homomorphism $\widehat\gamma^\sigma_{\rho^\vee}: \widehat{W}^{\mathrm{aff}}\map \{\pm 1\} $ by relation
\begin{equation}\label{dgammah}
 \widehat\gamma^\sigma_{\rho^\vee} (w^{\mathrm{aff}})= \widehat{\theta}_{\rho^\vee}(w^{\mathrm{aff}})\cdot [\sigma \circ \widehat\psi (w^{\mathrm{aff}})], \q w^{\mathrm{aff}} \in \widehat{W}^{\mathrm{aff}}.
\end{equation}
Note that, since for the dual affine reflection $r^\vee_0\in R^\vee$ it holds that $\widehat\psi(r^\vee_0)=r_\eta$, we have from \eqref{sleta} that 
\begin{equation}\label{dsigmar0}
 \sigma^s \circ \widehat\psi (r^\vee_0)=1, \q \sigma^l \circ \widehat\psi (r^\vee_0)=-1.
\end{equation}
The values of the $\widehat\gamma^\sigma_{\rho^\vee}-$homomorphism on any $w^{\mathrm{aff}} \in \widehat{W}^{\mathrm{aff}}$ are determined by its values on the generators from $R^\vee$. Putting together the values from  \eqref{dthetahom}, \eqref{ghomid} --- \eqref{sleta} and \eqref{dsigmar0}, we summarize the values of  $\widehat\gamma^\sigma_{\rho^\vee} (r),\,r\in R$ for a trivial $0-$shift and the non-trivial admissible shift $\rho^\vee$ in Table~\ref{Tab:dhom}.

\begin{table}
{\small \renewcommand{\arraystretch}{1.3}
\begin{tabular}{|c|c|c|c|c||c|c|c|c|}\hline 
                                    & $\widehat\gamma^{\id}_0$ & $\widehat\gamma^{\sigma^e}_0$  &  $\widehat\gamma^{\sigma^s}_0$ & $\widehat\gamma^{\sigma^l}_0$ & $\widehat\gamma^{\id}_{\rho^\vee}$  &  $\widehat\gamma^{\sigma^e}_{\rho^\vee}$ & $\widehat\gamma^{\sigma^s}_{\rho^\vee}$  &  $\widehat\gamma^{\sigma^l}_{\rho^\vee}$\\[1pt] \hline\hline
$r^\vee_0$                           & $ 1$ & $ -1$ & $ -1$ & $ 1$ & $ -1$ & $ 1$ & $ 1$ & $ -1$ \\[3pt]
$r_\al,\,\al\in\Delta_l$ & $ 1$ & $ -1$ &$ 1$ & $ -1$ & $ 1$ & $ -1$ & $ 1$ & $ -1$ \\[3pt]
$r_\al,\,\al\in\Delta_s$ & $ 1$ & $ -1$ & $ -1$ & $ 1$ & $ 1$ & $ -1$ & $ -1$ & $ 1$  \\[3pt]
\hline
\end{tabular}}\medskip
\caption{The values of  $\widehat\gamma^\sigma_{\rho^\vee} (r),\,r\in R^\vee$ for a trivial $0-$shift, the non-trivial admissible shift $\rho^\vee$ and four sign homomorphisms  $\sigma$.} \label{Tab:dhom}
\end{table}

\subsection{Fundamental domains}\

Any sign homomorphism and any admissible dual shift determine a decomposition of the fundamental domain~$F$. The factors of this decomposition are essential for the study of discrete orbit functions. For any sign homomorphism $\sigma$ and any admissible dual shift $\rho$ we introduce a subset $F^\sigma(\rho)$ of $F$
\begin{align*}
F^\sigma(\rho)&=\setb{a\in F}{ \gamma^\sigma_\rho\left( \mathrm{Stab}_{W^{\mathrm{aff}}}(a)\right)=\{1\} }
\end{align*}
with $\gamma^\sigma_\rho-$homomorphism given by \eqref{gammah}. Since for all points of the interior of $F$ the stabilizer is trivial, i.e. $\mathrm{Stab}_{W^{\mathrm{aff}}}(a)=1$, $a\in \mathrm{int} (F)$, the interior $\mathrm{int} (F)$ is a subset of any $F^\sigma(\rho)$. 
Let us also define the corresponding subset $R^\sigma(\rho)$ of generators $R$ of ${W}^{\mathrm{aff}}$
\begin{align}\label{Rsigma}
R^\sigma(\rho)&=\setb{r\in R}{ \gamma^\sigma_\rho\left( r\right)=-1 }
\end{align}
The sets $R^\sigma(\rho)$ are straightforwardly explicitly determined for any case from Table \ref{Tab:hom}.
In order to determine the analytic form of the sets $F^\sigma(\rho)$, we define subsets of the boundaries of $F$
\begin{align}\label{Hrho}
H^\sigma (\rho)&=\set{a\in F}{(\exists r\in R^\sigma(\rho))(ra=a)}.
\end{align}
\begin{tvr}\label{FsFl}For the sets $F^\sigma(\rho)$ it holds that
$$F^\sigma(\rho)=F\setminus H^\sigma (\rho).$$
\end{tvr}
\begin{proof}
Let $a\in F$. If $a\notin F\setminus H^\sigma (\rho)$, then $a\in H^\sigma (\rho)$, and there exists $r\in R^\sigma(\rho)$ such that $r\in \mathrm{Stab}_{W^{\mathrm{aff}}}(a)$. Then according to \eqref{Rsigma}, we have $\gamma^\sigma_\rho(r)=-1$. Thus, $\gamma^\sigma_\rho \left(\mathrm{Stab}_{W^{\mathrm{aff}}}(a)\right)=\{\pm 1\}$ and consequently $a\notin F^\sigma(\rho)$. Conversely, if $a\in F\setminus H^\sigma (\rho)$, then the stabilizer $\mathrm{Stab}_{W^{\mathrm{aff}}}(a)$ is either trivial or generated by generators from $R\setminus R^\sigma(\rho)$ only. Then, since for any generator $r\in R\setminus R^\sigma(\rho)$ it follows from \eqref{Rsigma} that $ \gamma^\sigma_\rho (r)=1$, we obtain $ \gamma^\sigma_\rho \left(\mathrm{Stab}_{W^{\mathrm{aff}}}(a)\right)=\{ 1\}$, i.e. $a\in F^\sigma(\rho)$.
\end{proof}
The explicit description of all domains $F^\sigma(\rho)$ follows from \eqref{deffun} and Proposition \ref{FsFl}. We define the symbols $y^{\sigma,\rho}_i\in\R$, $i=0,\dots,n$ in the following way:
\begin{equation*}
y^{\sigma,\rho}_i \in \begin{cases}\R^{> 0},\q   r_i\in R^\sigma(\rho) \\ \R^{\geq 0},\q r_i\in R\setminus R^\sigma(\rho).\end{cases}
\end{equation*}
Thus, the explicit form of $F^\sigma(\rho)$ is given by
\begin{equation}\label{FsFlex}
F^\sigma(\rho)=\setb{y^{\sigma,\rho}_1\om^{\vee}_1+\dots+y^{\sigma,\rho}_n\om^{\vee}_n}{y^{\sigma,\rho}_0+y^{\sigma,\rho}_1 m_1+\dots+y^{\sigma,\rho}_n m_n=1  }.
\end{equation}

\subsection{Dual fundamental domains}\

The dual version of the $\gamma-$homomorphism $\widehat\gamma^\sigma_{\rho^\vee}$ also determines a decomposition of the dual fundamental domain~$F^\vee$. The factors of this decomposition are essential for the study of the discretized orbit functions. We define subsets $F^{\sigma\vee}(\rho^\vee)$ of $F^\vee$ by 
\begin{align}\label{FsFldual}
F^{\sigma\vee}(\rho^\vee)&=\setb{a\in F^\vee}{ \widehat\gamma^\sigma_{\rho^\vee}\left( \mathrm{Stab}_{\widehat{W}^{\mathrm{aff}}}(a)\right)=\{1\} }
\end{align}
where $\widehat\gamma^\sigma_{\rho^\vee}$ is given by \eqref{dgammah} and $\rho^\vee$ is an admissible shift. Since for all points of the interior of $F^\vee$ the stabilizer is trivial, i.e. $\mathrm{Stab}_{\widehat{W}^{\mathrm{aff}}}(a)=1$, $a\in \mathrm{int} (F^\vee)$, the interior $\mathrm{int} (F^\vee)$ is a subset of all $F^{\sigma\vee}(\rho^\vee)$ . 
Let us also define the corresponding subset $R^{\sigma\vee}(\rho^\vee)$ of generators $R^\vee$ of $\widehat{W}^{\mathrm{aff}}$
\begin{align}\label{dRsigma}
R^{\sigma\vee}(\rho^\vee)&=\setb{r\in R^\vee}{ \widehat\gamma^\sigma_{\rho^\vee}\left( r\right)=-1 }
\end{align}
The sets $R^{\sigma\vee}(\rho^\vee)$ are straightforwardly explicitly determined for any case from Table \ref{Tab:dhom}.
We define subsets of the boundaries of $F^\vee$ by
\begin{align}\label{Hrhod}
H^{\sigma\vee}(\rho^\vee)&=\set{a\in F^\vee}{(\exists r\in R^{\sigma\vee}(\rho^\vee))(ra=a)}.
\end{align}
Similarly to Proposition \ref{FsFl} we derive its dual version.
\begin{tvr}\label{FsFld}For the sets $F^{\sigma\vee}(\rho^\vee)$ it holds that
$$F^{\sigma\vee}(\rho^\vee)=F^\vee\setminus H^{\sigma\vee}(\rho^\vee).$$
\end{tvr}
The explicit description of all domains $F^{\sigma\vee}(\rho^\vee)$ follows from \eqref{deffund} and Proposition \ref{FsFld}. We define the symbols $z^{\sigma,\rho^\vee}_i\in\R$, $i=0,\dots,n$ in the following way:
\begin{equation*}
z^{\sigma,\rho^\vee}_i \in \begin{cases}\R^{> 0},\q   r_i\in R^{\sigma\vee}(\rho^\vee) \\ \R^{\geq 0},\q r_i\in R^\vee\setminus R^{\sigma\vee}(\rho^\vee).\end{cases}
\end{equation*}
Thus, the explicit form of $F^{\sigma\vee}(\rho^\vee)$ is given by
\begin{equation}\label{dFsFlex}
F^{\sigma\vee}(\rho^\vee)=\setb{z^{\sigma,\rho^\vee}_1\om_1+\dots+z^{\sigma,\rho^\vee}_n\om_n}{z^{\sigma,\rho^\vee}_0+z^{\sigma,\rho^\vee}_1 m^\vee_1+\dots+z^{\sigma,\rho^\vee}_n m^\vee_n=1  }.
\end{equation}

\section{Four types of orbit functions}\
\subsection{Symmetries of orbit functions}\

Four sign homomorphisms $\id$, $\sigma^e$, $\sigma^s$ and $\sigma^l$ determine four types of families of complex orbit functions for root systems with two different lenghts of roots. 
If a given root system has only one length of roots there exist two homomorphism $\id$, $\sigma^e$ only.
Within each family of special functions, determined by a sign homomorphism $\sigma$, are the complex functions $\phi^\sigma_b:\R^n\map \C$ standardly labeled by weights $b\in P$. In this article we start by allowing $b\in \R^n$ and we have the orbit functions of the general form
\begin{equation*}
\phi^\sigma_b(a)=\sum_{w\in W}\sigma (w)\, e^{2 \pi i \sca{ wb}{a}},\q a\in \R^n.
\end{equation*}
The discretization properties of all four cases of these functions on a finite fragment
of the grid $\frac{1}{M}P^{\vee}$ and with $b\in P$ were described in \cite{HP,HMP}. 

Here we firstly examine the discretization by taking discrete values of $b\in \rho+P$, with $\rho$ an admissible dual shift. Using the $\gamma_\rho^\sigma-$homomorphism, defined by \eqref{gammah}, we are able to describe invariance properties of all types of orbit functions in the following compact form.
\begin{tvr}\label{diskval}
Let $\rho$ be an admissible dual shift and $b\in \rho+P$. Then for any $w^{\mathrm{aff}} \in {W}^{\mathrm{aff}}$ and $a\in \R^n$ it holds that 
\begin{equation}\label{Sssym}
\phi^\sigma_{b}(w^{\mathrm{aff}}a)=\gamma_\rho^\sigma(w^{\mathrm{aff}})\phi^\sigma_{b}(a).
\end{equation}
Moreover the functions $\phi^\sigma_{b}$ are zero  on the boundary $H^\sigma (\rho)$, i.e.
\begin{equation}\label{Fss}
\phi^\sigma_{b}(a')=0,\q  a'\in H^\sigma (\rho).
\end{equation}
\end{tvr}
\begin{proof}
Considering an element of the affine Weyl group of the form $w^{\mathrm{aff}}a=w'a+q^\vee$, $w'\in W$, $q^\vee\in Q^\vee$ and $b=\rho +\la$, with $\la\in P$, we firstly calculate that
\begin{align*}
 e^{2 \pi i \sca{ wb}{q^\vee}}&=e^{2 \pi i \sca{ \la+\rho}{w^{-1}q^\vee}}=e^{2 \pi i \sca{ \rho}{w^{-1}q^\vee}}= e^{2 \pi i \sca{ \rho}{q^\vee}}=\theta_\rho (w^{\mathrm{aff}}) 
\end{align*}
where the first equality follows from $w\in W$ being an orthogonal map, the second from $W-$invariance of $Q^\vee$ and $\Z-$duality of the lattices $P$ and $Q^\vee$, the third is equation \eqref{shiftdexp}. 
Then we have that 
\begin{align*}
\phi^\sigma_{b}(w^{\mathrm{aff}}a)&=\sum_{w\in W}\sigma (w)\, e^{2 \pi i \sca{ wb}{w'a+q^\vee}}=\sum_{w\in W}\sigma (w)\, e^{2 \pi i \sca{ wb}{w'a}}  e^{2 \pi i \sca{ wb}{q^\vee}}\\
   &=\theta_\rho (w^{\mathrm{aff}}) \sum_{w\in W}\sigma (w)\, e^{2 \pi i \sca{ wb}{w'a}}= \theta_\rho (w^{\mathrm{aff}}) \sigma(w')\phi^\sigma_{b}(a)
\end{align*}
which is equation \eqref{Sssym}. Putting the generators \eqref{Rsigma} from $R^\sigma(\rho)$ into \eqref{Sssym} we immediately have for the points from $H^\sigma(\rho)$ that 
\begin{align*}
\phi^\sigma_{b}(a)&=\phi^\sigma_{b}(ra)=-\phi^\sigma_{b}(a), \q r\in  R^\sigma(\rho), \q a\in H^\sigma(\rho) .
\end{align*}
\end{proof}
Thus, all orbit functions are (anti)symmetric with respect to the action of the affine Weyl group~$W^{\mathrm{aff}}$. This allows us to consider the values $\phi^\sigma_{b}(a)$ only for points of the fundamental domain $a\in F$. Moreover,
from \eqref{Fss} and Proposition \ref{FsFl} we conclude that {\it we can consider the functions $\phi^\sigma_{b}$, $b\in \rho+P$  on the  domain $F^\sigma(\rho)$ only}.

Secondly, reverting to a general $b\in \R^n$, we discretize the values of $a$ by taking $a\in \frac{1}{M}(\rho^\vee +P^{\vee})$  with $\rho^\vee$ an admissible shift and any $M\in \N$. 
Using the $\widehat\gamma_{\rho^\vee}^\sigma-$homomorphism, defined by \eqref{dgammah}, we describe invariance with respect to the action of $\widehat{W}^{\mathrm{aff}} $ in the following compact form.
\begin{tvr}\label{disklab} 
Let $\rho^\vee$ be an admissible shift and $a\in \frac{1}{M}(\rho^\vee +P^{\vee})$ with $M\in \N$. Then for any $w^{\mathrm{aff}} \in \widehat{W}^{\mathrm{aff}}$ and $b\in \R^n$ it holds that 
\begin{equation}\label{dSssym}
\phi^\sigma_{Mw^{\mathrm{aff}}(\frac{b}{M})}(a)=\widehat\gamma_{\rho^\vee}^\sigma(w^{\mathrm{aff}})\phi^\sigma_{b}(a).
\end{equation}
Moreover the functions $\phi^\sigma_{b}$ are identically zero on the boundary $MH^{\sigma \vee}(\rho^\vee)$, i.e.
\begin{equation}\label{dFss}
\phi^\sigma_{b}\equiv 0,\q  b\in MH^{\sigma \vee}(\rho^\vee).
\end{equation}
\end{tvr}
\begin{proof}
Considering an element of the dual affine Weyl group of the form $w^{\mathrm{aff}}a=w'a+q$ and $a=(\rho^\vee +s)/M$, with $s\in P^\vee$, we firstly calculate that
\begin{align*}
 e^{2 \pi i \sca{ qM}{wa}}&=e^{2 \pi i \sca{w^{-1} q}{\rho^\vee+s}}=e^{2 \pi i \sca{ \rho^\vee}{w^{-1}q}}= e^{2 \pi i \sca{ \rho^\vee}{q}}=\widehat\theta_{\rho^\vee} (w^{\mathrm{aff}}) 
\end{align*}
where the first equality follows from $w\in W$ being an orthogonal map, the second from $W-$invariance of $Q$ and $\Z-$duality of the lattices $P^\vee$ and $Q$, the third is equation \eqref{shiftexp}. 
Then we have that 
\begin{align*}
\phi^\sigma_{Mw^{\mathrm{aff}}(\frac{b}{M})}(a)&=\sum_{w\in W}\sigma (w)\, e^{2 \pi i \sca{ w'b+qM}{wa}}=\sum_{w\in W}\sigma (w)\, e^{2 \pi i \sca{ w'b}{wa}}  e^{2 \pi i \sca{ qM}{wa}}\\
   &=\widehat\theta_{\rho^\vee} (w^{\mathrm{aff}}) \sum_{w\in W}\sigma (w)\, e^{2 \pi i \sca{ w'b}{wa}}= \widehat\theta_{\rho^\vee} (w^{\mathrm{aff}}) \sigma(w')\phi^\sigma_{b}(a)
\end{align*}
which is equation \eqref{dSssym}. Putting the generators \eqref{dRsigma} from $R^{\sigma\vee}(\rho^\vee)$ into \eqref{dSssym} we immediately have for the points from $MH^{\sigma \vee}(\rho^\vee)$ that 
\begin{align*}
\phi^\sigma_{b}(a)&=\phi^\sigma_{Mr(\frac{b}{M})}(a)=-\phi^\sigma_{b}(a), \q r\in  R^{\sigma\vee}(\rho^\vee), \q b\in MH^{\sigma \vee}(\rho^\vee) .
\end{align*}
\end{proof}
Thus, all orbit functions are (anti)symmetric with respect to the action of the modified dual affine Weyl group --- the group $T(M Q) \rtimes W$. This allows us to consider the functions $\phi^\sigma_{b}$ only for labels from the fundamental domain $b\in MF^\vee$. Moreover,
from \eqref{dFss} and Proposition \ref{FsFld} we conclude that {\it we can consider the functions $\phi^\sigma_{b}(a)$, $a\in \frac{1}{M}(\rho^\vee +P^{\vee})$  with the labels $b$ from $MF^{\sigma\vee}(\rho^\vee)$ only}.

\subsection{Discretization of orbit functions}\

By discretization of orbit functions we mean formulating their discrete Fourier calculus; at this point this step is straightforward --- we discretize both the arguments of the functions $\phi^\sigma_b(a)$ to $\frac{1}{M}(\rho^\vee +P^{\vee})$ and the labels to $\rho+P$ with $\rho$ and $\rho^\vee$ being admissible. Combining propositions \ref{diskval} and \ref{disklab}, we obtain that the discrete values $a$ of the functions $\phi^\sigma_b(a)$ can be taken only on the set
\begin{equation}\label{Fs}
F^\sigma_M (\rho,\rho^\vee)= \left[\frac{1}{M}(\rho^\vee+P^{\vee})\right]\cap F^\sigma(\rho)
\end{equation}
and the labels $b$ are then considered in the set 
\begin{equation}\label{Ls}
\Lambda^\sigma_M(\rho,\rho^\vee) =  (\rho+P) \cap MF^{\sigma\vee}(\rho^\vee).
\end{equation}
The explicit form of these sets is crucial for application and can be for each case of admissible shifts $\rho^\vee=\rho^\vee_1\om_1^\vee+\dots + \rho^\vee_n\om_n^\vee $ and admissible dual shifts $\rho=\rho_1\om_1+\dots + \rho_n\om_n $ derived from \eqref{FsFlex} and \eqref{dFsFlex}. If we define the  symbols 
\begin{equation*}
u^{\sigma,\rho}_i \in \begin{cases}\N,\q   r_i\in R^\sigma(\rho) \\ \Z^{\geq 0},\q r_i\in R\setminus R^\sigma(\rho)\end{cases}
\end{equation*}
and the symbols
\begin{equation*}
t^{\sigma,\rho^\vee}_i \in \begin{cases}\N,\q   r_i\in R^{\sigma\vee}(\rho^\vee) \\ \Z^{\geq 0},\q r_i\in R^\vee\setminus R^{\sigma\vee}(\rho^\vee)\end{cases}
\end{equation*}
we obtain the set $F^\sigma_M (\rho,\rho^\vee)$ explicitly
\begin{align}\label{FMrho}
F^\sigma_M (\rho,\rho^\vee) &=\setb{\frac{u^{\sigma,\rho}_1+\rho^\vee_1}{M}\om^{\vee}_1+\dots+\frac{u^{\sigma,\rho}_n+\rho^\vee_n}{M}\om^{\vee}_n}{u^{\sigma,\rho}_0+u^{\sigma,\rho}_1 m_1+\dots+u^{\sigma,\rho}_n m_n=M  }
\end{align}
and also the set $\Lambda^\sigma_M(\rho,\rho^\vee)$
\begin{align}\label{LMrho}
\Lambda^\sigma_M (\rho,\rho^\vee) &=\setb{(t^{\sigma,\rho^\vee}_1+\rho_1)\om_1+\dots+(t^{\sigma,\rho^\vee}_n+\rho_n)\om_n}{t^{\sigma,\rho^\vee}_0+t^{\sigma,\rho^\vee}_1 m^\vee_1+\dots+t^{\sigma,\rho^\vee}_n m^\vee_n=M  }.
\end{align}
The counting formulas for the numbers of elements of all $F^\sigma_M (0,0)$ are derived in \cite{HP,HMP}. Recall from \cite{HP} that the number of elements of $|F^\id_M(0,0)|$ of $C_n,\, n\geq 2$ is equal to $\nu_n (M)$, where $\nu_n (M)$ is given on odd and even values of $M$ by,
\begin{align*}
\nu_n (2k)=& \begin{pmatrix}n+k \\ n \end{pmatrix}+\begin{pmatrix}n+k-1 \\ n \end{pmatrix}\\
\nu_n (2k+1)= & 2\begin{pmatrix}n+k \\ n \end{pmatrix} .
\end{align*}

\begin{thm}\label{numAn}
Let $\rho^\vee$, $\rho$  be the non-trivial admissible and the dual admissible shifts, respectively.  The numbers of points of grids $F^\sigma_M$ of Lie algebras $A_1$, $B_n$, $C_n$ are given by the following relations.
\begin{enumerate}\item  $A_1$,
\begin{align*}
|F^{\id}_M (\rho,0)|= & |F^{\id}_M (0,\rho^\vee)|=|F^{\id}_M (\rho,\rho^\vee)|=M-1, \\
|F^{\sigma^e}_M (\rho,0)|= & |F^{\sigma^e}_M (0,\rho^\vee)|=|F^{\sigma^e}_M (\rho,\rho^\vee)|=M-1.
\end{align*}
\item $C_2,$
\begin{align*}
|F^{\id}_M (\rho,0)|= & |F^{\id}_M (0,\rho^\vee)|=|F^{\sigma^s}_M (0,\rho^\vee)|=|F^{\sigma^l}_M (\rho,0)|=\nu_2 (M-1) \\
|F^{\id}_M (\rho,\rho^\vee)|= & |F^{\sigma^e}_M (\rho,\rho^\vee)|=|F^{\sigma^s}_M (\rho,\rho^\vee)|=|F^{\sigma^l}_M  (\rho,\rho^\vee)|=\nu_2 (M-2),\\
|F^{\sigma^e}_M (0,\rho^\vee)|= & |F^{\sigma^e}_M (\rho,0)|=|F^{\sigma^s}_M (\rho,0)|=|F^{\sigma^l}_M  (0,\rho^\vee)|=\nu_2 (M-3),
\end{align*}
\item $C_n$, $n\geq 3,$
\begin{align*}
|F^{\id}_M (\rho,0)|& =|F^{\sigma^l}_M (\rho,0)|= \nu_n (M-1) \\
|F^{\sigma^e}_M (\rho,0)|& =|F^{\sigma^s}_M (\rho,0)|=\nu_n (M-2n+1)
\end{align*}
\item $B_n$, $n\geq 3,$
\begin{align*}
|F^{\id}_M (0,\rho^\vee)|& =|F^{\sigma^s}_M (0,\rho^\vee)|= \nu_n (M-1) \\
|F^{\sigma^e}_M (0,\rho^\vee)|& =|F^{\sigma^l}_M (0,\rho^\vee)|=\nu_n (M-2n+1).
\end{align*}
\end{enumerate}
\end{thm}
\begin{proof}
The number of points is calculated for each case of admissible shifts using formula \eqref{FMrho}.
For the case $C_n$, we have from Table \ref{Tab:hom} that $R^{\sigma^s}(\rho)=\{r_0,r_1,\dots, r_{n-1} \}$ and thus
\begin{align}\label{numCn}
F^{\sigma^s}_M (\rho,0) &=\setb{\frac{u^{\sigma^s,\rho}_1}{M}\om^{\vee}_1+\dots+\frac{u^{\sigma^s,\rho}_n}{M}\om^{\vee}_n}{u^{\sigma^s,\rho}_0+2u^{\sigma^s,\rho}_1 +\dots+2u^{\sigma^s,\rho}_{n-1}+u^{\sigma^s,\rho}_n =M  }
\end{align}
where
\begin{equation*}
u^{\sigma^s,\rho}_i \in \begin{cases}\N,\q   i\in \{0,1,\dots,n-1\}, \\ \Z^{\geq 0},\q i=n. \end{cases}
\end{equation*}
Introducing new variables $\wt u^{\sigma^s,\rho}_i\in \Z^{\geq 0}$ and setting 
\begin{equation*}
u^{\sigma^s,\rho}_i = \begin{cases}\wt u^{\sigma^s,\rho}_i+1,\q   i\in \{0,1,\dots,n-1\}, \\ \wt u^{\sigma^s,\rho}_i,\q i=n. \end{cases}
\end{equation*}
the defining equation in the set \eqref{numCn} can be rewritten as 
\begin{equation*}
\wt u^{\sigma^s,\rho}_0+2\wt u^{\sigma^s,\rho}_1 +\dots+2\wt u^{\sigma^s,\rho}_{n-1}+\wt u^{\sigma^s,\rho}_n =M-2n+1.	
\end{equation*}
This equation is the same as the defining equation of $F^{\id}_{M-2n+1} (0,0)$ of $C_n$, hence 
\begin{equation}\label{countCn}
|F^{\sigma^s}_M (\rho,0)|	=|F^{\id}_{M-2n+1} (0,0)|
\end{equation}
and the counting formula follows. Similarly, we obtain formulas for the remaining cases.
\end{proof}
\begin{example}\label{ex1}
For the Lie algebra $C_2$, we have its determinant of the Cartain matrix $c=2$. For $M=4$ is the order of the group $\frac{1}{4}P^{\vee}/Q^{\vee}$ equal to $32$, and according to Theorem \ref{numAn} we calculate 
\begin{align*}
|F^{\id}_4 (\rho,0)|= & |F^{\id}_4 (0,\rho^\vee)|=|F^{\sigma^s}_4 (0,\rho^\vee)|=|F^{\sigma^l}_4 (\rho,0)|=6 \\
|F^{\id}_4 (\rho,\rho^\vee)|= & |F^{\sigma^e}_4 (\rho,\rho^\vee)|=|F^{\sigma^s}_4 (\rho,\rho^\vee)|=|F^{\sigma^l}_4  (\rho,\rho^\vee)|=4,\\
|F^{\sigma^e}_4 (0,\rho^\vee)|= & |F^{\sigma^e}_4 (\rho,0)|=|F^{\sigma^s}_4 (\rho,0)|=|F^{\sigma^l}_4  (0,\rho^\vee)|=2,
\end{align*}
Let us denote the boundaries of the triangle $F$ which are stabilized by the reflection $r_i$ by $b_i$, $i\in\{0,1,2\}$. Then the sets $H^\sigma (\rho)$ of boundaries \eqref{Hrho} are of the explicit form 
\begin{equation*}
\begin{alignedat}{3}
H^\id (\rho)&=\{b_0\}, &\q H^{\sigma^e} (\rho)&=\{b_1,b_2\},  \\
H^{\sigma^s} (\rho)&=\{b_0,b_1\}, &\q H^{\sigma^l} (\rho)&=\{b_2\}.
\end{alignedat}
\end{equation*}
The coset representatives of the finite group $\frac{1}{4}P^{\vee}/Q^{\vee}$, the shifted representatives of the set of cosets $\frac{1}{4}(\rho^\vee + P^{\vee})/Q^{\vee}$  and the fundamental domain $F$ with its boundaries $\{b_0,b_1,b_2\}$ are depicted in Figure~\ref{figC2}.
\begin{figure}
\includegraphics{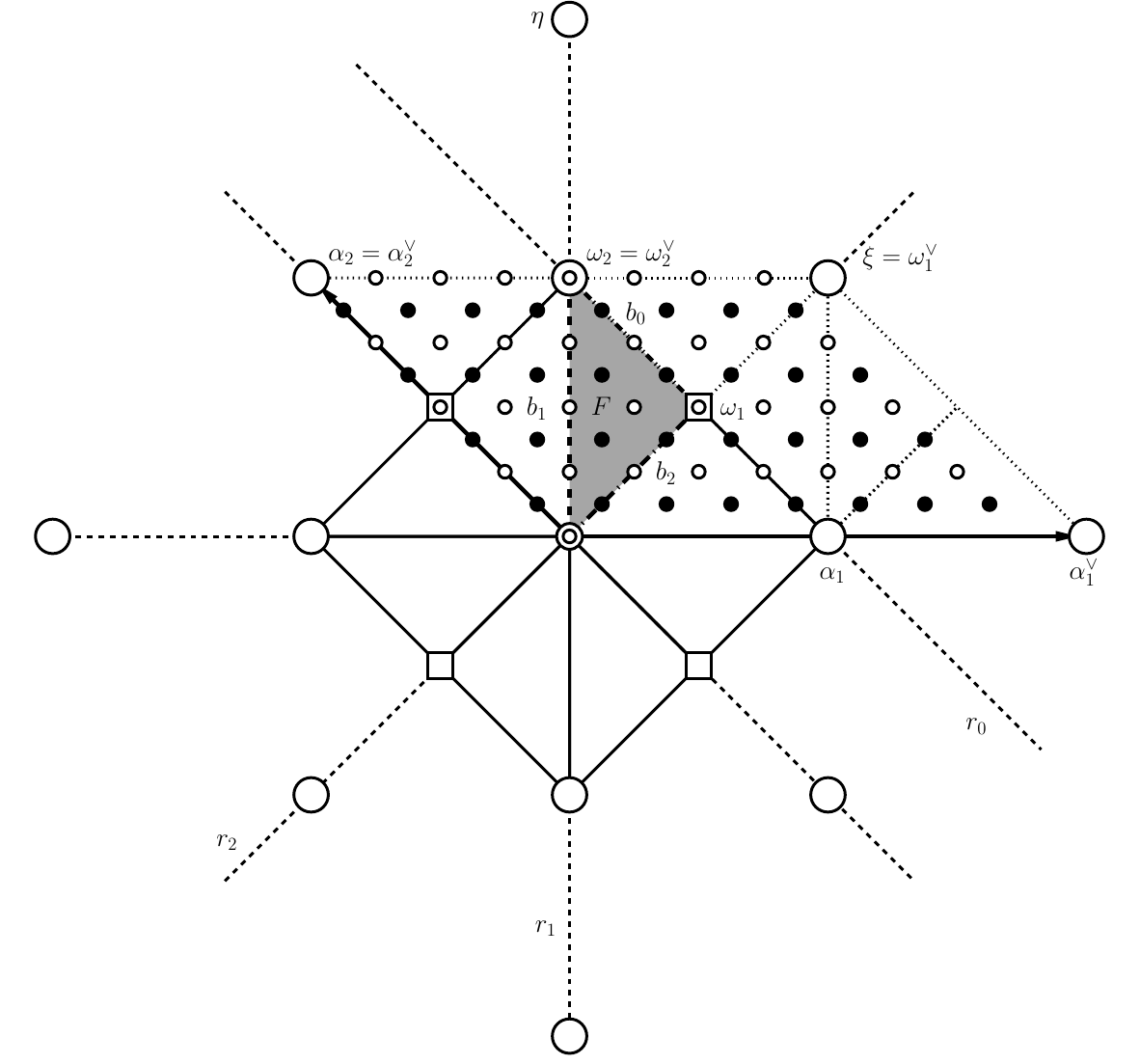}
\caption{\small The fundamental domain $F$ and its boundaries of $C_2$. The fundamental domain $F$ is depicted as the grey triangle containing borders  $b_1$, $b_2$ and $b_0$, which are depicted as the thick dashed line,  dot-and-dashed and dot-dot-and-dashed lines, respectively. The coset representatives of $\frac{1}{4}P^{\vee}/Q^{\vee}$ and the shifted representatives of cosets $\frac{1}{4}(\rho^\vee + P^{\vee})/Q^{\vee}$ are shown as $32$ white and black dots, respectively. The six black points in $F$ form the set $F^{\id}_4 (0,\rho^\vee)$. The dashed lines represent 'mirrors' $r_0,r_1$ and $r_2$. Larger circles are elements of the root lattice $Q$; together with the squares they are elements of the weight lattice~$P$.}\label{figC2}
\end{figure}
\end{example}

It is established in \cite{HP,HMP} that the numbers of elements of $F^\sigma_M (0,0)$ coincide with the numbers of elements of $\Lambda^\sigma_M (0,0)$. Here we generalize these results to include all cases of admissible shifts.
\begin{thm}\label{numAnd}
Let $\rho^\vee$, $\rho$ be any admissible and any dual admissible shifts, respectively. Then it holds that
\begin{equation*}
|\Lambda^\sigma_M(\rho,\rho^\vee)|=|F^\sigma_M(\rho,\rho^\vee)|	.
\end{equation*}
\end{thm}
\begin{proof}
The equality $|\Lambda^\sigma_M(0,0)|=|F^\sigma_M(0,0)|$ is proved in \cite{HP,HMP}. For other combinations of admissible shifts is the number of points calculated for each case using formula \eqref{LMrho}.
For the case $C_n$, we have from Table \ref{Tab:dhom} that $R^{\sigma^s \vee}(0)=\{r^\vee_0,r_1,\dots, r_{n-1} \}$ and thus
\begin{align}\label{numCnd}
\Lambda^{\sigma^s}_M (\rho,0) &=\setb{t^{\sigma^s,0}_1\om_1+\dots+\left(t^{\sigma^s,0}_n+\frac{1}{2}\right)\om_n}{t^{\sigma^s,0}_0+t^{\sigma^s,0}_1 +2t^{\sigma^s,0}_{2}+\dots+2\left(t^{\sigma^s,0}_n+\frac{1}{2}\right) =M  }
\end{align}
where
\begin{equation*}
t^{\sigma^s,0}_i \in \begin{cases}\N,\q   i\in \{0,1,\dots,n-1\}, \\ \Z^{\geq 0},\q i=n. \end{cases}
\end{equation*}
Introducing new variables $\wt t^{\sigma^s,\rho}_i\in \Z^{\geq 0}$ and setting 
\begin{equation*}
t^{\sigma^s,0}_i = \begin{cases}\wt t^{\sigma^s,0}_i+1,\q   i\in \{0,1,\dots,n-1\}, \\ \wt t^{\sigma^s,\rho}_i,\q i=n. \end{cases}
\end{equation*}
the defining equation in the set \eqref{numCnd} can be rewritten as 
\begin{equation*}
\wt t^{\sigma^s,0}_0+\wt t^{\sigma^s,0}_1 +2\wt t^{\sigma^s,0}_{2}+\dots+2\wt t^{\sigma^s,0}_n =M-2n+1.	
\end{equation*}
This equation is the same as the defining equation of $\Lambda^{\id}_{M-2n+1} (0,0)$ of $C_n$ and taking into account \eqref{countCn} we obtain that $$|\Lambda^{\sigma^s}_M (\rho,0)|=|\Lambda^{\id}_{M-2n+1} (0,0)|=|F^{\id}_{M-2n+1} (0,0)|=|F^{\sigma^s}_M (\rho,0)|.$$ Similarly, we obtain the equalities for the remaining cases.
\end{proof}

\begin{example}\label{exdual}
For the Lie algebra $C_2$ and for $M=4$ is the order of the group $P/4Q$ equal to $32$, and according to Theorem \ref{numAnd} we calculate 
\begin{align*}
|\Lambda^{\id}_4 (\rho,0)|= & |\Lambda^{\id}_4 (0,\rho^\vee)|=|\Lambda^{\sigma^s}_4 (0,\rho^\vee)|=|\Lambda^{\sigma^l}_4 (\rho,0)|=6 \\
|\Lambda^{\id}_4 (\rho,\rho^\vee)|= & |\Lambda^{\sigma^e}_4 (\rho,\rho^\vee)|=|\Lambda^{\sigma^s}_4 (\rho,\rho^\vee)|=|\Lambda^{\sigma^l}_4  (\rho,\rho^\vee)|=4,\\
|\Lambda^{\sigma^e}_4 (0,\rho^\vee)|= & |\Lambda^{\sigma^e}_4 (\rho,0)|=|\Lambda^{\sigma^s}_4 (\rho,0)|=|\Lambda^{\sigma^l}_4  (0,\rho^\vee)|=2,
\end{align*}
Let us denote the boundaries of the triangle $4F^\vee$ which are stabilized by the reflection $r^\vee_i$ by $b^\vee_i$, $i\in\{0,1,2\}$. Then the sets $4H^{\sigma^\vee} (\rho^\vee)$ of boundaries \eqref{Hrhod} are the relevant boundaries of $4F^\vee$ and are given as  \begin{equation*}
\begin{alignedat}{3}
4H^{\id \vee} (\rho^\vee)&=\{b^\vee_0\}, &\q 4H^{\sigma^e \vee} (\rho^\vee)&=\{b^\vee_1,b^\vee_2\},  \\
4H^{\sigma^s \vee} (\rho^\vee)&=\{b^\vee_1\}, &\q 4H^{\sigma^l \vee} (\rho^\vee)&=\{b^\vee_0,b^\vee_2\}.
\end{alignedat}
\end{equation*}
The coset representatives of the finite group $P/4Q$, the shifted representatives of the set of cosets $(\rho+P)/4Q $  and the magnified fundamental domain $4F^\vee$ with its boundaries $\{b^\vee_0,b^\vee_1,b^\vee_2\}$ are depicted in Figure~\ref{figC2d}.
\begin{figure}
\includegraphics{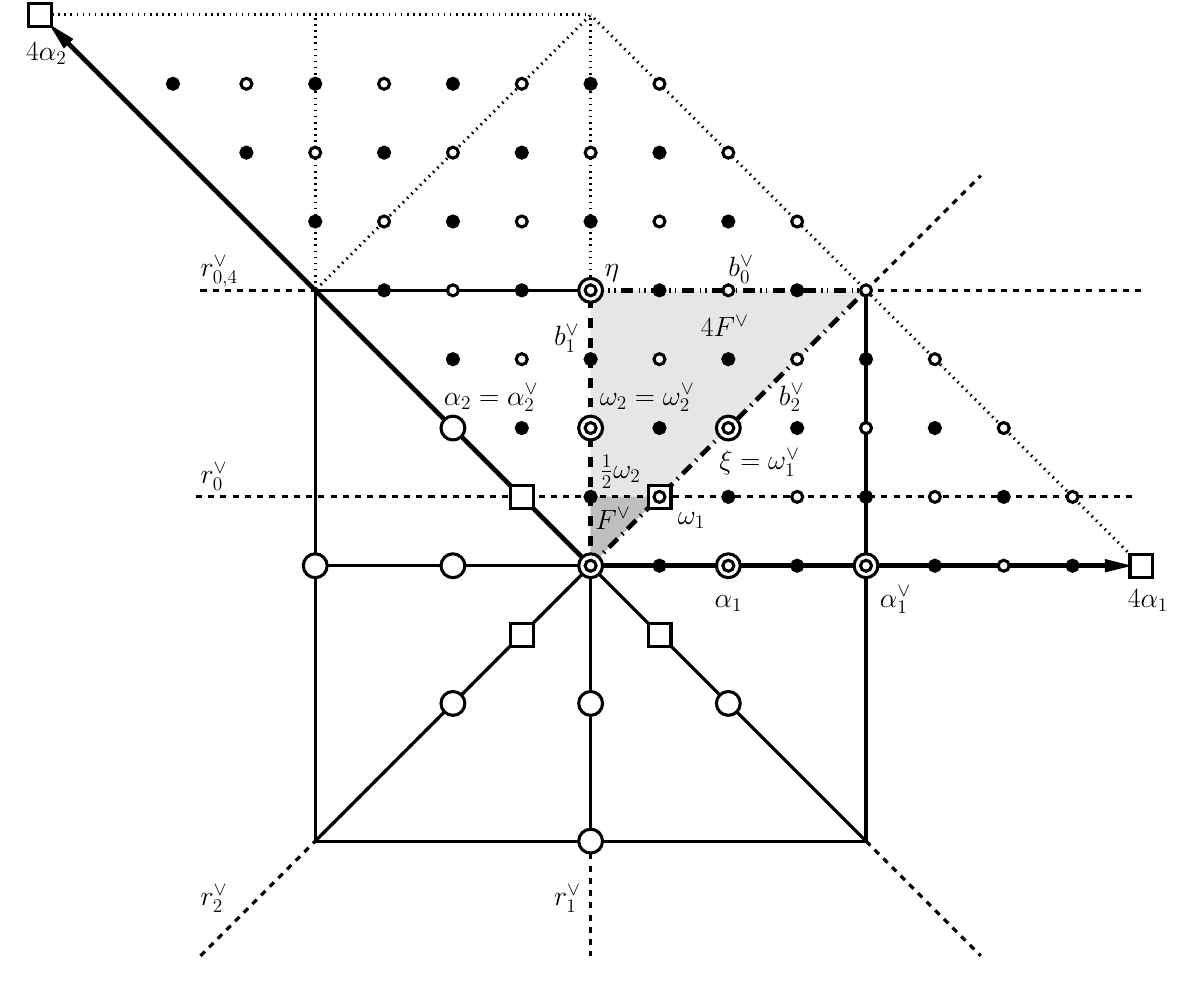}
\caption{\small The magnified fundamental domain $4F^\vee$ and its boundaries of $C_2$. The darker grey triangle is the fundamental domain $F^{\vee}$ and the lighter grey triangle is the domain $4F^{\vee}$ containing borders $b^\vee_1$, $b^\vee_2$ and $b^\vee_0$, which are depicted as the thick dashed line,  dot-and-dashed and dot-dot-and-dashed lines, respectively. The coset representatives of $P/4Q$ and the shifted representatives of cosets $(\rho+P)/4Q $ are shown as $32$ white and black dots, respectively. The six black points in $4F^\vee$ form the set $\Lambda^{\id}_4 (\rho,0)$. The dashed lines represent 'mirrors' $r^\vee_0,r^\vee_1$ and $r^\vee_2$ and the affine mirror $r^\vee_{0,4}$ is defined by $r^\vee_{0,4}\la= 4r^\vee_0(\la/4)$. The larger circles and squares coincide with those in Figure \ref{figC2}.}\label{figC2d}
\end{figure}
\end{example}

\section{Discrete orthogonality and transforms of orbit functions}
\subsection{Discrete orthogonality of orbit functions}\

To describe the discrete orthogonality of the $\phi^\sigma_b$ functions with shifted points from $F^\sigma_M (\rho,\rho^\vee)$ and shifted weights from $\Lambda^\sigma_M (\rho,\rho^\vee)$, we need to generalize the discrete orthogonality from \cite{MP2, HP, HMP} . Recall that basic discrete orthogonality relations of the exponentials from \cite{HP,MP2} imply for any $\mu \in P$ that
\begin{equation}\label{bdis}
 \sum_{y\in\frac{1}{M}P^{\vee}/Q^{\vee}} e^{2\pi i\sca{\mu}{y}}=\begin{cases} cM^n,\q \mu\in MQ  \\ 0,\q\q \mu\notin MQ.\end{cases}
\end{equation}
The scalar product of any two functions $f,g:F^\sigma_M (\rho,\rho^\vee)\map \Com$ is given by
\begin{equation} \label{scp}
\sca{f}{g}_{F^\sigma_M (\rho,\rho^\vee)}= \sum_{a\in F^\sigma_M (\rho,\rho^\vee)}\ep(a) f(a)\overline{g(a)},
\end{equation}
where the numbers $\ep (a)$ are determined by (\ref{epR}). We prove that $\Lambda^\sigma_M (\rho,\rho^\vee)$ is the lowest maximal set of pairwise orthogonal orbit functions.
\begin{thm}\label{orthothm}
Let $\rho^\vee$, $\rho$ be any admissible and any dual admissible shifts, respectively, and $\sigma$ be a sign homomorphism. Then for any $b,b' \in\Lambda_M^\sigma (\rho,\rho^\vee)$ it holds that
\begin{equation}\label{ortho}
 \sca{\phi^\sigma_b}{\phi^\sigma_{b'}}_{F^\sigma_M (\rho,\rho^\vee)}=c\abs{W}M^n h^{\vee}_M(b) \,\delta_{b,b'},
\end{equation}
where $c$, $h^{\vee}_M(b)$ were defined by (\ref{Center}), (\ref{hMb}),
respectively, $|W|$ is the number of elements of the Weyl group.
\end{thm}
\begin{proof}
Since due to \eqref{Fss} any $\phi^\sigma_b$ vanishes on $H^\sigma (\rho)$ and it holds that
$$F^\sigma_M (\rho,\rho^\vee) \cup  \left[\frac{1}{M}(\rho^\vee+P^{\vee})\cap H^\sigma (\rho) \right]  =  \left[\frac{1}{M}(\rho^\vee+P^{\vee})\right]\cap F, $$ we have
\begin{equation*}
\sca{\phi^\sigma_b}{\phi^\sigma_{b'}}_{F^\sigma_M (\rho,\rho^\vee)}= \sum_{a\in F^\sigma_M (\rho,\rho^\vee)}\ep(a) \phi^\sigma_b(a)\overline{\phi^\sigma_{b'}(a)}=\sum_{a\in \left[\frac{1}{M}(\rho^\vee+P^{\vee})\right]\cap F}\ep(a) \phi^\sigma_b(a)\overline{\phi^\sigma_{b'}(a)}
\end{equation*}
Combining the relations \eqref{epshift} and \eqref{Sssym}, we observe that the expression $\ep(a) \phi^\sigma_b(a)\overline{\phi^\sigma_{b'}(a)}$ is ${W}^{\mathrm{aff}}-$invariant; the shift invariance with respect to $Q^\vee$, together with \eqref{ept}, implies that 
$$\sum_{a\in \left[\frac{1}{M}(\rho^\vee+P^{\vee})\right]\cap F}\ep(a) \phi^\sigma_b(a)\overline{\phi^\sigma_{b'}(a)}= \sum_{x\in \left[\frac{1}{M}(\rho^\vee+P^{\vee})/Q^\vee\right]\cap F}\wt\ep(x) \phi^\sigma_b(x)\overline{\phi^\sigma_{b'}(x)}
$$
and taking into account also $W-$invariance of $\wt\ep(x) \phi^\sigma_b(x)\overline{\phi^\sigma_{b'}(x)}$ together with \eqref{rfun1} and \eqref{rfun2}, we obtain
\begin{equation*}
\sum_{x\in \left[\frac{1}{M}(\rho^\vee+P^{\vee})/Q^\vee\right]\cap F}\wt\ep(x) \phi^\sigma_b(x)\overline{\phi^\sigma_{b'}(x)}= \sum_{y \in \frac{1}{M}(\rho^\vee+P^{\vee})/Q^\vee} \phi^\sigma_b(y)\overline{\phi^\sigma_{b'}(y)}.
\end{equation*}
Using the $W-$invariance of $\frac{1}{M}(\rho^\vee+P^{\vee})/Q^\vee$, which follows from \eqref{rhoa}, we obtain
\begin{align}
\sca{\phi^\sigma_b}{\phi^\sigma_{b'}}_{F^\sigma_M (\rho,\rho^\vee)} = &\sum_{w'\in W}\sum_{w\in W} \sum_{y\in \frac{1}{M}(\rho^\vee+P^{\vee})/Q^\vee}\sigma(ww')e^{2\pi\i\sca{wb-w'b'}{y}}\nonumber \\=&\abs{W}\sum_{w'\in W}\sum_{y \in \frac{1}{M}(\rho^\vee+P^{\vee})/Q^\vee}\sigma(w')e^{2\pi\i\sca{b-w'b'}{y}} \nonumber\\ 
=&\abs{W}\sum_{w'\in W}\sigma(w')e^{2\pi\i\sca{b-w'b'}{\frac{\rho^\vee}{M}}}\sum_{y \in \frac{1}{M}P^{\vee}/Q^\vee}e^{2\pi\i\sca{b-w'b'}{y}}. \nonumber
\end{align}
Having the labels $b$ and $b'$ of the form $b=\rho+\la$ and $b'=\rho+\la'$, with $\la,\la'\in P$, we obtain from Proposition \ref{Wdreq} that
$$ b-w'b' = \rho- w'\rho +\la -\la'\in P.
$$
Thus, the basic orthogonality relation \eqref{bdis} can be used. Taking into account that \eqref{dfun2} together with $b,b'\in MF^\vee$ and $b-w'b' \in MQ $, for some $w'\in W$, forces $b=b'$, we observe that
if $b\neq b'$ then for all $w'\in W$ it has to hold that $ b-w'b' \notin MQ$. We conclude that if $b\neq b'$ then \eqref{bdis} implies $\sca{\phi^\sigma_b}{\phi^\sigma_{b'}}_{F^\sigma_M (\rho,\rho^\vee)}=0$. On the other hand if $b=b'$ then \eqref{bdis} forces all summands  for which $ b-w'b \notin MQ$, i.e. $w' \notin \widehat\psi(\mathrm{Stab}_{\widehat{W}^{\mathrm{aff}}}\left(\frac{b}{M}\right)) $ to vanish. Therefore we finally obtain
\begin{align}
\sca{\phi^\sigma_b}{\phi^\sigma_{b}}_{F^\sigma_M (\rho,\rho^\vee)} 
=& c\abs{W}M^n\sum_{w'\in \widehat\psi\left[\mathrm{Stab}_{\widehat{W}^{\mathrm{aff}}}\left(\frac{b}{M}\right)\right]}\sigma(w')e^{2\pi\i\sca{b-w'b}{\frac{\rho^\vee}{M}}}\label{usehom} \\ =&
c\abs{W}M^n\sum_{w^{\mathrm{aff}}\in \mathrm{Stab}_{\widehat{W}^{\mathrm{aff}}}\left(\frac{b}{M}\right)}\sigma(\widehat\psi[w^{\mathrm{aff}}] )e^{2\pi\i\sca{\widehat\tau[w^{\mathrm{aff}}]}{\rho^\vee}} \nonumber\\ = & c\abs{W}M^n\sum_{w^{\mathrm{aff}}\in \mathrm{Stab}_{\widehat{W}^{\mathrm{aff}}}\left(\frac{b}{M}\right)} \widehat\gamma^\sigma_{\rho^\vee} (w^{\mathrm{aff}}) \nonumber
\end{align} 
and taking into account \eqref{FsFldual} and definition \eqref{hMb}, we have
$$\sum_{w^{\mathrm{aff}}\in \mathrm{Stab}_{\widehat{W}^{\mathrm{aff}}}\left(\frac{b}{M}\right)} \widehat\gamma^\sigma_{\rho^\vee} (w^{\mathrm{aff}})=h^{\vee}_M(b).$$ 
\end{proof}

\begin{example}
Consider the root system of $C_2$ ---
its highest root $\xi$ and the highest dual root $\eta$, depicted in Figure \ref{figC2}, are given by
$$\xi=2\al_1+\al_2,\ \ \eta=\al^\vee_1+2\al^\vee_2.$$
The Weyl group has eight elements, $|W|=8$, and the determinant of the Cartan matrix is $c=2$. We fix a sign homomorphism to $\sigma= \sigma^s$ and
take a parameter with coordinates in $\omega-$basis $(a,b)$ and a point with coordinates in $\alpha^\vee$-basis $(x,y)$. Then the $S^{s}-$functions are of the following explicit form \cite{HMP}, 
\begin{align*}
\phi^{\sigma^s}_{(a,b)}(x,y)= & 2\{ \cos(2\pi((a+2b)x-by))+\cos(2\pi(ax+by)) \\ &-\cos(2\pi((a+2b)x-(a+b)y))-\cos(2\pi(ax-(a+b)y))\}.
\end{align*}
Fixing both admissible shifts to the non-zero values from Table \ref{Tab:shifts}, the grid $F^{\sigma^s}_M (\rho,\rho^\vee)$ is of the form
\begin{align}\label{FC2}
F^{\sigma^s}_M (\rho,\rho^\vee) =& \setb{\frac{u_1+\frac{1}{2}}{M}\om^{\vee}_1+\frac{u_2}{M}\om^{\vee}_2}{u_1,\,u_2\in \Z^{\geq 0},\, u_0\in\N,\, u_0+2u_1+u_2=M} 
\end{align}
and the grid $\Lambda^{\sigma^s}_M (\rho,\rho^\vee)$ is of the form 
\begin{align}\label{LC2}
\Lambda^{\sigma^s}_M (\rho,\rho^\vee)=& \setb{t_1\om_1+\left(t_2+\frac{1}{2}\right)\om_2}{t_0, t_2\in \Z^{\geq 0}, t_1\in\N,\, t_0+t_1+2t_2=M}.
\end{align}
For calculation of the coefficients $\ep(x)$, $h^{\vee}_M(b)$, which appear in \eqref{scp} and \eqref{ortho}, a straightforward generalization of the calculation procedure, described in \S 3.7 in \cite{HP}, is used. Each point $a\in F^{\sigma^s}_M (\rho,\rho^\vee) $ and each label $b\in \Lambda^{\sigma^s}_M (\rho,\rho^\vee)$ are assigned the coordinates $u_0,u_1,u_2$ and $t_0,t_1,t_2$ from \eqref{FC2} and \eqref{LC2}, respectively, and the triples which enter the algorithm in \cite{HP} are $[u_0,u_1+\frac{1}{2},u_2]$ and $[t_0,t_1,t_2+\frac{1}{2}]$. Thus, the values $\ep(a)$ of the point represented by $u_0,u_1,u_2$ are $\ep(a)=8$ if $u_2\neq 0$ and $\ep(a)=4$ if $u_2=0$. The values $h^\vee_M(b)$ of the point represented by $t_0,t_1,t_2$ are $h^\vee_M(b)=1$ if $t_0\neq 0$ and $h^\vee_M(b)=2$ if $t_0=0$ and therefore the relations orthogonality \eqref{ortho} are of the form 
\begin{equation*}
 \sca{\phi^{\sigma^s}_b}{\phi^{\sigma^s}_{b'}}_{F^{\sigma^s}_M (\rho,\rho^\vee)}=16M^2 h^{\vee}_M(b) \,\delta_{b,b'}.
\end{equation*}
\end{example}

\subsection{Discrete trigonometric transforms}\

Similarly to ordinary Fourier analysis, interpolating functions $I^\sigma_M (\rho,\rho^\vee)$
\begin{align}
I^\sigma_M (\rho,\rho^\vee)(a)= \sum_{b\in \Lambda^{\sigma}_M (\rho,\rho^\vee)} c^{\sigma}_b \phi^{\sigma}_b(a),\q a\in \R^n \label{intc}
\end{align}
are given in terms of expansion functions $\phi^\sigma_b$ and unknown expansion coefficients $c^{\sigma}_b$.
For some function $f$ sampled on the grid $F^{\sigma^s}_M (\rho,\rho^\vee)$, the interpolation of $f$ consists of finding the coefficients $c^{\sigma}_b$ in the interpolating functions (\ref{intc}) such that it coincides with $I^\sigma_M (\rho,\rho^\vee)$ at all gridpoints, i.e.  
\begin{align*}
I^\sigma_M (\rho,\rho^\vee)(a)=& f(a), \q a\in F^{\sigma}_M (\rho,\rho^\vee).
\end{align*}
The coefficients $c^{\sigma}_b$ are due to Theorem \ref{numAnd} uniquely determined via the standard methods of calculation of Fourier coefficients
\begin{align}\label{trans}
c^{\sigma}_b=& \frac{\sca{f}{\phi^\sigma_b}_{F^{\sigma}_M (\rho,\rho^\vee)}}{\sca{\phi^\sigma_b}{\phi^\sigma_b}_{F^{\sigma}_M (\rho,\rho^\vee)}}=(c\abs{W} M^n h^{\vee}_M(b))^{-1}\sum_{a\in F^{\sigma}_M (\rho,\rho^\vee)}\ep(a) f(a)\overline{\phi^\sigma_b(a)}
\end{align}
and the corresponding Plancherel formulas are of the form
\begin{align*}
\sum_{a\in F^{\sigma}_M (\rho,\rho^\vee)} \ep(a)\abs{f(a)}^2 = & c \abs{W} M^n \sum_{b\in\Lambda^{\sigma}_M(\rho,\rho^\vee)}h^{\vee}_M(b)|c^\sigma_b|^2 .
\end{align*}

\section{Concluding Remarks}
\begin{itemize}
\item The discrete orthogonality relations of the orbit functions in Theorem \ref{orthothm} and the corresponding discrete transforms \eqref{trans} contain as a special case a compact formulation of the previous results --- definition \eqref{Fs} corresponds for zero shifts to the definitions of the four sets $F_M, \wt F_M, F^s_M, F^l_M$ in \cite{HP,HMP} and the same holds for the sets of weights. Except for Theorem \ref{numAnd}, the proofs are also developed in a uniform form.  However, the approach to the proof of Theorem \ref{numAnd}, which states the crucial fact of the completeness of the obtained sets of functions, still relies on the case by case analysis of the numbers of points which is done for trivial shifts in \cite{HP,HMP}. 
\item To any complex function $f$, defined on $F^\sigma(\rho)$, is by relations \eqref{intc} and \eqref{trans} assigned a functional series $\{I^\sigma_M (\rho,\rho^\vee)\}_{M=1}^\infty$.  
Existence of conditions for convergence of these functional series together with an estimate of the interpolation error $\int |f -I^\sigma_M (\rho,\rho^\vee) |^2$ poses an open problem.
\item The four standard discrete cosine transforms DCT--I,..., DCT--IV and the four sine transforms DST--I,..., DST--IV from \cite{Brit} are included as a special cases of \eqref{trans} corresponding to the algebra $A_1$ with its two sign homomorphisms and two admissible shifts. The case of $C_2$ appears to have exceptionally rich outcome of the shifted transforms --- the twelve new transforms will be detailed in a separated article.  Note that Chebyshev polynomials of one variable of the third and fourth kinds \cite{Hand} are induced using the dual admissible shifts of $A_1$. This indicates that similar families of orthogonal polynomials might be generated for higher-rank cases.   
\item Since the four types of orbit functions generate special cases of Macdonald polynomials, it can be expected that the orthogonality relations \eqref{ortho}, parametrized by two independent admissible shifts, will translate into a generalization of discrete polynomial orthogonality relations from \cite{DE}. 
The most ubiquitous of the four types of orbit functions is the family of $\phi^{\sigma^e}-$functions --- these antisymmetric orbit functions are constituents of the Weyl character formula and generate Schur polynomials. Their discrete orthogonality, discussed e. g. in \cite{Kir,Kac,D}, is also generalized for the admissible cases. Discrete orthogonality relations induce the possibility of deriving so called cubatures formulas \cite{MP4}, which evaluate exactly multidimensional integrals of polynomials of a certain degree, and generates polynomial interpolation methods \cite{RM}. Especially in the case of $C_2$ cubature formulas, which are developed in \cite{xuc2}, several new such formulas may be expected.    
\item Another family of special functions, which can be investigated for having similar shifting properties, is the set of so called $E-$functions \cite{HP2,KP3}. In this case, the condition \eqref{rhoa} needs to be weakened to the even subgroup of $W$ only and can be expected to produce discrete Fourier analysis on new types of grids; indeed this requirement is trivial for the case of $A_1$ and even allows an arbitrary shift. 
\item A natural question arises if the restriction \eqref{rhoa} could be weakened to obtain even richer outcome of the shifted transforms for the present case. For the analysis on the dual weights grid $P^\vee$ this, however, does not seem to be straightforwardly possible --- condition of $W-$invariance \eqref{rhoa} is equivalent to the existence of the shift homomorphism \eqref{shifthom}, which enters the definitions of the fundamental domain. By inducing sign changes the shift homomorphism controls the boundary conditions on the affine boundary of the fundamental domain and allows to evaluate expression \eqref{usehom}. Thus, the approach in the present work may serve as a starting point for further research on developing the discrete Fourier analysis of orbit functions on different types of grids.     
\end{itemize}

\section*{Acknowledgments}
The authors are grateful for partial support by the project 7AMB13PL035-8816/R13/R14. JH gratefully acknowledges the support of this work by RVO68407700.

\end{document}